\documentclass[12pt]{article}

\usepackage{amsmath}
\usepackage{amssymb}
\usepackage{amsthm}
\usepackage{array}
\usepackage{IEEEtrantools}
\usepackage{mathrsfs}
\usepackage[round, longnamesfirst]{natbib}
\usepackage{thmtools}
\usepackage{thm-restate}
\usepackage{enumitem}
\usepackage{url}
\usepackage{lipsum}
\usepackage{graphicx}
\usepackage{wrapfig}
\usepackage{subcaption}
\usepackage{mathabx}
\usepackage{comment}
\usepackage[noend]{algpseudocode}
\usepackage{bm}
\usepackage{cancel}

\newlength{\eqboxstorage}

\newcounter{theorem-counter}

\theoremstyle{plain}
\newtheorem*{theorem*}{Theorem}
\newtheorem*{corollary*}{Corollary}
\newtheorem*{lemma*}{Lemma}
\newtheorem{theorem}[theorem-counter]{Theorem}
\newtheorem{lemma}{Lemma}

\newtheorem{proposition}{Proposition}
\newtheorem{example}{Example}
\newtheorem*{example*}{Example}

\theoremstyle{definition}
\newtheorem{definition}{Definition}

\theoremstyle{remark}

\newtheorem{remark}{Remark}

\DeclareMathOperator{\argmin}{argmin}

\newcommand{\da}{\downarrow}
\newcommand{\tran}[0]{\text{tran}}
\newcommand{\ti}{\tilde}

\newcommand{\bs}[1]{\boldsymbol{#1}}
\newcommand{\tbs}[1]{\tilde{\bs{#1}}}
\newcommand{\hbs}[1]{\hat{\bs{#1}}}
\newcommand{\bbs}[1]{\bar{\bs{#1}}}

\newcommand{\swaps}[0]{\text{Swaps}}

\newcommand{\ls}[0]{\text{LS}}
\newcommand{\rlp}[0]{\mathbb{R}_+^L}

\setlist{nolistsep}

\usepackage[onehalfspacing]{setspace}

\title{ Goodness-of-fit and utility estimation: \\ what's possible and what's not}
\date{\vspace{-5ex}}
\author{Yujian Chen, Joshua Lanier, and John K.-H. Quah\footnote{Chen: School of Economics and Center for Intelligence Economic Science, Southwestern University of Finance and Economics; Lanier: Division of Economics, Nanyang Technological University; Quah: Department of Economics, National University of Singapore. Lanier gratefully acknowledges financial support from the National Natural Science Foundation of China (W2433187).}}
\begin{document}
\maketitle


\begin{abstract}
A goodness-of-fit index measures the consistency of consumption data with a given model of utility-maximization.  We show that for the class of well-behaved (i.e., continuous and increasing) utility functions there is no goodness-of-fit index that is continuous and {\em accurate}, where the latter means that a perfect score is obtained if and only if a dataset can be rationalized by a well-behaved utility function. While many standard goodness-of-fit indices are inaccurate we show that these indices are (in a sense we make precise) \emph{essentially accurate}. Goodness-of-fit indices are typically generated by loss functions and we find that standard loss functions usually do not yield a best-fitting utility function when they are minimized. Nonetheless, welfare comparisons can be made by working out a {\em robust preference} relation from the data.
\end{abstract}


\section{Introduction}

The well-known theorem of Afriat (see \citet{afriat67}, \citet{diewert73}, and \citet{varian82}) provides economists with a way to test whether a dataset consisting of purchasing decisions from a consumer is consistent with utility-maximization. This test, while easily implementable and widely applied, can be overly demanding. Indeed, researchers have found that both observational and experimental datasets often fail the test in its exact form.  In other words, it is common for consumers or experimental subjects to fall short of exact utility-maximizing behavior. For this reason, researchers have developed various goodness-of-fit indices that measure the severity of a dataset's departure from utility-maximization. In this paper, we provide a systematic examination of goodness-of-fit indices and their close relation, loss functions.  We also propose a criterion for welfare comparisons based on loss functions, and provide an empirical demonstration of our welfare criterion using supermarket scanner data. 

Let $D$ be a dataset consisting of a consumer's purchasing decisions. $D$ consists of $T$ observations where, at observation $t$, the consumer is observed to buy a bundle $\bs{q}^t\in\mathbb{R}_+^L$ of $L$ goods at prices $\bs{p}^t\in\mathbb{R}^L_{++}$; thus, we may write $D$ as $( \bs{q}^t, \bs{p}^t )_{1\leq t \leq T}$.  We say that $D$ can be rationalized by, or is consistent with, the maximization of a well-behaved utility function $U:\mathbb{R}^L_+\to\mathbb{R}$ (where {\em well-behaved} means that $U$ is continuous and increasing) if the function $U$ generates $\bs{q}^t$ as demand at prices $\bs{p}^t$, for every observation $t$; more formally, at each observation $t$, we have  $U(\bs{q}^t)\geq U(\bs{q})$ for all bundles $\bs{q}$ such that $\bs{p}^t \cdot \bs{q} \leq \bs{p^t}\cdot\bs{q}^t$. Afriat's Theorem states that a dataset $D$ can be rationalized by a well-behaved utility function if and only if the dataset satisfies the generalized axiom of revealed preferences (GARP), where GARP is a condition on chosen bundles that excludes revealed preference cycles (see Section \ref{sec:bad-news} for a precise statement). 

A {\em goodness-of-fit index} measures the severity of a dataset's departure from exact utility-maximization. Such an index $\tau$ is a map from $\cal D$ (the collection of all datasets with finitely many observations) to the non-negative numbers. A lower goodness-of-fit measurement (i.e.\ a lower value of $\tau(D)$) indicates that the dataset is closer to being perfectly consistent with well-behaved utility maximization, with $\tau (D)=0$ whenever $D$ is exactly consistent with well-behaved utility maximization.   


We focus on five popular goodness-of-fit indices in this paper. The Afriat index, from \citet{afriat73}, measures goodness-of-fit by measuring the fraction of the consumer's budget which is misspent. The Varian index, from \citet{varian90}, is similar in spirit to the Afriat index but provides a more nuanced measure of budgetary waste. The Swaps index, from \citet{apesteguia-ballester15}, measures (using the Lebesgue measure) the portion of the budget set which is strictly preferred to the bundle actually purchased. The nonlinear Least Squares index measures the distance between the purchase data and the nearest GARP-satisfying dataset. The Houtman-Maks index, from \citet{houtman1985}, reports the number of observations which must be discarded in order for the purchase data to satisfy GARP.\vspace{0.15in} 

\noindent {\bf Tension Between Accuracy and Continuity.}\; It seems sensible to require the goodness-of-fit index $\tau$ to be a continuous function of the data\footnote{We can think of $D = (\bs{q}^t, \bs{p}^t)_{t \leq T}$ as an element of $\mathbb{R}^{2T}$ and so there is a natural way to define continuity for goodness-of-fit indices (see Definition \ref{def:goodness-of-fit-props}).} and also {\em accurate}, in the sense that the index outputs $0$ if and only if $D$ is consistent with the maximization of some well-behaved utility function. An index which is not continuous can give drastically different reports for nearly identical datasets whereas an index which is not accurate will sometimes give the wrong impression about whether the data can be explained by well-behaved utility-maximization.    

The first result of this paper (Proposition \ref{prop:badnews2}) states that there are no goodness-of-fit indices that are continuous and accurate. Thus, all the goodness-of-fit indices used in the literature are either discontinuous or inaccurate. The reason for this is, in essence, that the datasets which obey GARP do not form a closed set. Indeed, one could construct a sequence of datasets $D_n$ obeying GARP with a limit $\widebar D$ where GARP is violated (see Example \ref{example:bad-news} in the next section); accuracy demands $\tau (D_n)=0$ for all $n$ and $\tau(D) >0$ but continuity requires $\tau(D) = \lim_n \tau (D_n) = 0$ and thus $\tau$ cannot be both accurate and continuous. Let us refer to these GARP-violating datasets which perch on the boundary of the rational datasets as \emph{cusp datasets}. 

A tension exists over how to treat cusp datasets and their neighbors. Continuous indices which assign zero to GARP-satisfying datasets must also assign a zero to a cusp dataset. Moreover, continuity forces datasets near to the cusp to take values close to zero (even if they violate GARP). On the other hand, accurate indices assign a cusp dataset (and nearby datasets that also violate GARP) with a strictly positive value. Now, if an economist is performing some empirical revealed preference analysis and many of the datasets investigated are near the cusp then her results will be quite sensitive to whether she uses an accurate or a continuous index.

Among the commonly used indices, we show that the Afriat, Varian, Swaps, and Least Squares indices are continuous but not accurate while the Houtman-Maks index is accurate but discontinuous. That these goodness-of-fit indices have shortcomings is somewhat understood in the literature. Indeed, \citet{andreoni-miller02} and \citet{murphy-banerjee15} point out that the Afriat index is inaccurate and \citet{halevypersitzzrill18} (in footnote 10 of their paper) point out that the Varian index is inaccurate. Similarly, \citet{fisman07} (in appendix III) point out that the Houtman-Maks index has a continuity problem. What has been less understood is that the inaccuracy of the Afriat and Varian indices and the discontinuity of the Houtman-Maks index is not a defect of these particular goodness-of-fit indices but rather the outcome of an inevitable trade-off between accuracy and continuity.\vspace{0.15in} 


\noindent {\bf Essential Accuracy.}\;  While it is too much to ask that an index be continuous and accurate we show that continuity is compatible with a weakened version of accuracy. We say that a goodness-of-fit measure $\tau$ is {\em essentially accurate} if, whenever $\tau(D)=0$, then the purchased bundles in $D$ can be perturbed by an arbitrarily small amount so that the perturbed dataset is consistent with well-behaved utility maximization. We show that the Afriat, Varian, Swaps, and Least Squares indices all satisfy essential accuracy. While this result can mitigate concerns that these indices are inaccurate it does not negate the central tension faced when choosing an index: datasets near a cusp dataset will be treated in a very different fashion depending on whether a continuous or an accurate index is used and there seems to be no principled way to determine which type of index is preferable.   \vspace{0.15in}


\noindent {\bf Loss functions.}\;  A loss function $Q(U;D)$ assigns a non-negative number to each utility function $U$ and purchase dataset $D$. Intuitively, $Q(U;D)$ reports how well the choices in $D$ can be described as being the outcome of maximizing $U$. Loss functions can be used, as in Halevy et al. (2018), to estimate the utility function of the consumer by using $\argmin_{U} Q(U;D)$.  That is, the estimator is a utility function which provides the best description of the data according to $Q$. There is a natural relationship between loss functions and goodness-of-fit indices: a loss function $Q$ generates a goodness-of-fit index via 
$$\tau(D) = \inf_{U} Q(U;D)$$ 
where the infimum is taken over all well-behaved utility functions. It seems natural to want a loss function $Q$ to (i)\, be {\em accurate} in the sense that it returns a perfect score of $0$ if and only if $D$ is rationalized by $U$; (ii)\, be {\em minimizable}, in the sense that $\argmin_U Q(U;D)$ is non-empty; and (iii)\, generate a continuous goodness-of-fit measure. Proposition \ref{prop:bad-news2} shows that there is no loss function satisfying these properties. 

It turns out that the Afriat, Varian, Swaps, Least squares, and Houtman-Maks indices could all be generated by loss functions. Their corresponding loss functions are accurate but none of them are minimizable, except for the Houtman-Maks loss function.\vspace{0.15in}


\noindent {\bf Robust Preference Relation.}\; The absence of a loss-minimizing utility function leads to complications for welfare analysis. Specifically, we might like to determine if a consumer with dataset $D$ prefers bundle $\tbs{q}$ to $\tbs{q}'$. If the collection of best-fitting utility functions, i.e., the set $\argmin_U Q(U;D)$, is nonempty, then it is natural to reach this conclusion if every $U \in \argmin_U Q(U;D)$ ranks $\tbs{q}$ above $\tbs{q}'$. It is less obvious how one is to make statements of this type when a best-fitting utility function does not exist. 

We think that there is a criterion that both makes sense and is empirically implementable.  We say that the bundle $\tbs{q}$ is {\em robustly preferred} to another bundle $\tbs{q}'$ if there is a well-behaved utility function $\ti{U}$ so that $U(\tbs{q}) \geq U(\tbs{q}')$ for every well-behaved utility function $U$ which provides a better fit than $\ti{U}$ in the sense that $Q(U;D) \leq Q(\ti{U};D)$. We show that this robust preference relation has basic properties we would expect of a good welfare measure: it is reflexive and transitive and is computationally feasible for the Afriat and Varian loss functions.\vspace{0.15in} 

\noindent {\bf Other classes of utility functions.}\; There are other classes of utility functions worth examining besides well-behaved utility functions.  For example, it is worth asking if a dataset $D$ can be rationalized by a homothetic utility function; in the case where consumption is over contingent commodities, it makes sense to ask if $D$ can be rationalized by a utility function with the expected utility form.  And if $D$ cannot be rationalized by a utility function from a given class, we would be interested in measuring the size and nature of its departure from that class. In Proposition \ref{prop:good-news} we characterize those classes of utility functions $\cal U$ that admit accurate and minimizable loss functions which generate accurate and continuous goodness-of-fit indices; in essence the characterizing condition is that the collection of datasets which admit a rationalization by a member of $\cal U$ must form a closed set.  We show that the class of homothetic utility functions and the class of utility functions with the expected utility form (with concave Bernoulli functions) satisfy this property.\vspace{0.15in}  

\noindent {\bf Bounded rationality.}\; In many models of bounded rationality, an agent has a utility function but does not optimize perfectly because of some explicitly modelled cognitive limitation. In this case, a loss function $Q(U;D)$ can be understood as a measure of the extent to which the boundedly rational process must be invoked in order to explain how an agent with utility function $U$ could make the choices observed in $D$. Indeed, it is possible to define a loss function that incorporates a specific boundedly rational model, which we illustrate in the paper with consideration set models (the approach is generalized and applied to other models of bounded rational behavior in the online appendix.).  This loss function can then be employed to calculate robust preferences, etc, in a way that is consistent with that model of bounded rationality. \vspace{15pt}


\noindent {\bf Empirical application.}\; We undertake an empirical demonstration of the robust preference concept by investigating the welfare consequences of a reduction in soda consumption for a panel of US households. In our demonstration we calculate how much non-soda beverage consumption must be increased in order to compensate households for a 25\% decrease in soda consumption. We find broad agreement in the amount of compensation required across the three indices we consider: Afriat, Varian, and Houtman Maks. Nevertheless, we find that the Varian index consistently provides sharper welfare bounds than the other two indices. \vspace{0.15in}

\noindent {\bf Related literature.}\; Apart from the goodness-of-fit indices we have already introduced, other examples include the money pump index of \citet{echenique11}, the minimum cost index of \citet{dean-martin16}, the FOC-departure index of \citet{deClippel23}, and the perturbations index of  \citet{hu25}.  There appears to have been little in the way of a broad investigation of the properties of goodness-of-fit indices along the lines of this paper.  In \citet{apesteguia-ballester15}, the Swaps index is axiomatized in a finite choice space context; two of the axioms presented there are continuity and accuracy (which they refer to as ``rationality"). While the finite choice set version of the Swaps index is accurate, it is inaccurate when the choice space is  $\mathbb{R}_+^L$.\footnote{Another approach to making sense of GARP violations is to test for consistency with the hypothesis of random utility maximization; examples in this literature include \citet{kitamura-stoye18} and \citet{deb-kitamura-quah-stoye22}. However, it makes sense to measure a consumer's rationality even in this model.  A natural idea is that a consumer is closer to well-behaved utility maximization than another if her utility functions are less stochastic (in some sense that needs to be formalized). 

We also mention that revealed preference conditions are used to make counterfactual predictions and welfare statements in \citet{varian82}, \citet{blundell-browning-crawford03}, \citet{cherchye-demuynck-derock19}, and \citet{adams23}.}     

Using the argmin of a loss function as an estimator is a standard procedure. Ordinary least squares, maximum likelihood, and the generalized method of moments are instances of this approach. In \citet{apesteguia-ballester15} the Swaps loss function is introduced as a means to estimate preferences (note that most of this paper assumes a finite choice space). \citet{chambers-echenique-lambert21} consider how to recover preferences from inconsistent binary choice data (i.e.\ agents are given pairs of alternatives and asked to select one of them); their approach uses a version of the Kemeny distance as a loss function. \citet{halevypersitzzrill18} explores using the Afriat, Varian, Least Squares, and Houtman-Maks loss functions to estimate the utility function of a consumer but avoids the problems we point out here by confining utility functions to a parametric family. \vspace{0.15in}


\noindent {\bf Organization of paper.}\;  Section \ref{sec:bad-news} introduces goodness-of-fit indices and explains why these indices for the class of well-behaved  utility functions cannot be both accurate and continuous. Section \ref{sec:bad-news2} focuses on loss functions. We state an impossibility result for loss functions and show that the Afriat, Varian, Swaps, and Least Squares loss functions are not minimizable for the class of well-behaved utility functions. This section also develops the robust preference criterion for welfare analysis and discusses boundedly rational models. Section \ref{sec:other} analyzes other classes of utility functions and characterizes those classes that admit accurate and minimizable loss functions that generate accurate and continuous goodness-of-fit indices. Section \ref{sec:empirical} presents our empirical demonstration and Section \ref{sec:conclusion} concludes. An appendix collects the proofs which do not appear in the body of the paper and provides a method for calculating the Afriat index exactly. An online appendix contains a section which discusses boundedly rational models, a section which discussed how to compute the robust preference relation for the Houtman-Maks loss function, and a section which provides a characterization of the robust preference relation for the Varian loss function. 

\section{Utility maximization and goodness-of-fit indices} \label{sec:bad-news}

We depict consumption bundles of $L$ goods by a vector $\bs{q} = (q_1, q_2, \ldots, q_L) \in \rlp$, where $q_{\ell} \geq 0$ is the amount of good $\ell$ consumed. A vector $\bs{p} = (p_1, p_2, \ldots, p_L) \in \mathbb{R}_{++}^L$ is a price vector where $p_{\ell} > 0$ is the price of good $\ell$. The analyst observes $T$ purchasing decisions of a consumer and assembles these decisions into a \emph{purchase dataset} $D = ( \bs{q}^t, \bs{p}^t )_{t \leq T}$ where $\bs{q}^t \in \rlp \backslash \{ \bs{0} \}$ is the bundle purchased by the consumer in period $t$ when the prices of the $L$ goods were $\bs{p}^t \in \mathbb{R}_{++}^L$. The purchase dataset $D$ is \emph{rationalized} by a utility function $U: \rlp \rightarrow \mathbb{R}$ if $U$ explains the data, in the sense that, at each observation $t$, $U(\bs{q}^t) \geq U(\bs{q})$ for all $\bs{q} \in B(\bs{p}^t, \bs{p}^t \cdot \bs{q}^t)$, where 
\begin{equation} \label{eq:budget-set}
	B(\bs{p}, m) = \{ \bs{q} \in \mathbb{R}_+^L: \bs{p} \cdot \bs{q} \leq m \}
\end{equation}
is the linear budget set arising from price $\bs{p}$ and expenditure level $m\in\mathbb{R}_+$; in other words, $\bs{q}^t$ gives more utility than any other bundle that is weakly cheaper than it at prevailing prices.  A utility function which is increasing\footnote{By increasing we mean $\bs{q} > \bs{q}'$ implies $U(\bs{q}) > U(\bs{q}')$.} and continuous is said to be \emph{well-behaved}; we denote the family of such utility functions by ${\cal U}_{WB}$.  If a purchase dataset can be rationalized by a well-behaved utility function, we refer to such a dataset as \emph{rationalizable}.


Let $D = ( \bs{q}^t, \bs{p}^t )_{t \leq T}$ be a purchase dataset. If $\bs{q}$ was affordable when $\bs{q}^t$ was purchased, i.e.\ $\bs{p}^t \cdot \bs{q}^t \geq \bs{p}^t \cdot \bs{q}$, then we write $\bs{q}^t \ R \ \bs{q}$; further, if $\bs{q}$ was strictly cheaper than $\bs{q}^t$ when $\bs{q}^t$ was purchased (i.e.\ $\bs{p}^t \cdot \bs{q}^t > \bs{p}^t \cdot \bs{q}$) then we write $\bs{q}^t \ P \ \bs{q}$. We refer to the binary relations $R$ and $P$ as the {\em direct} and {\em strict direct revealed preference relations} for $D$. The dataset $D$ satisfies the \emph{generalized axiom of revealed preference} (GARP) if, for all $t_1, t_2, \ldots, t_K$ such that $\bs{q}^{t_1} \ R \ \bs{q}^{t_2} \ R \ \ldots \ R \ \bs{q}^{t_K}$, it is not the case that $\bs{q}^{t_K} \ P \ \bs{q}^{t_1}$. In other words, $D$ violates GARP when it contains a {\em revealed preference-cycle}, or {\em RP-cycle} for short, i.e. a cycle of revealed preference relations
\begin{equation}\label{violateGARP}
	\bs{q}^{t_1} \ R \ \bs{q}^{t_2} \ R \ \ldots \ R \ \bs{q}^{t_K} \ R \ \bs{q}^{t_1}
\end{equation} 
where at least one of the $R$ relations can be replaced with $P$. The following result is due to \citet{afriat67}.

\begin{theorem}[Afriat] \label{theorem:afriat}
	A purchase dataset $D = (\bs{q}^t,\bs{p}^t)_{t \leq T}$ can be rationalized by a well-behaved utility function if and only if $D$ satisfies GARP.
\end{theorem}

\subsection{Goodness-of-fit indices} \label{sec:goodness-of-fit}

A \emph{goodness-of-fit index} is a function $\tau$ which maps each purchase dataset to a non-negative number. Intuitively, a smaller value of $\tau(D)$ means that the behavior exhibited in the dataset is more in line with the model of utility maximization being tested. Below we define a goodness-of-fit index and identify two of its desirable properties.

\begin{definition} \label{def:goodness-of-fit-props}
	Let $\mathcal{U}$ be some non-empty collection of utility functions. A map $\tau$ from purchase datasets to $\mathbb{R}_+$ is a {\em goodness-of-fit index for $\cal U$} if $\tau (D)=0$ whenever $D$ can be rationalized by some element of $\cal U$.  The index $\tau$ is \emph{accurate for $\mathcal{U}$} if  $\tau(D) = 0$ if and only if $D$ can be rationalized by some $U \in \mathcal{U}$ and $\tau$ is continuous if, for all $T$, the expression $\tau( (\bs{q}^t, \bs{p}^t)_{t \leq T} )$ is continuous as a function of the data $(\bs{q}^t, \bs{p}^t)_{t \leq T}$.\footnote{In other words, for each $T$ the goodness-of-fit index $\tau$ is continuous when treated as a function which maps $( ( \rlp \backslash \{ \bs{0} \} ) \times \mathbb{R}_{++}^L)^T$ into $\mathbb{R}_{+}$. Note that a purchase dataset $D = (\bs{q}^t, \bs{p}^t)_{t \leq T}$ has, by definition, $\bs{q}^t \neq \bs{0}$ for each $t$ which is why a continuous $\tau$ is thought of as a mapping from $( ( \rlp \backslash \{ \bs{0} \} ) \times \mathbb{R}_{++}^L)^T$ into $\mathbb{R}_{+}$ instead of $( \rlp \times \mathbb{R}_{++}^L)^T$ into $\mathbb{R}_{+}$.} 
\end{definition}

Accurate goodness-of-fit indices correctly report when a dataset is consistent with the tested hypothesis and correctly report when a dataset is not consistent. Continuous goodness-of-fit indices do not report drastically different numbers for infinitesimal changes in the data.  Unfortunately, the  following example demonstrates that continuity and accuracy are not compatible properties for a goodness-of-fit index, in the case where ${\cal U}={\cal U}_{WB}$.\footnote{For other class of utilty functions, these properties can be compatible; see Section \ref{sec:other}.}  

\begin{figure}[h]
	\centering
	\begin{subfigure}[b]{0.45\textwidth}
		\centering
		\includegraphics[width=\textwidth]{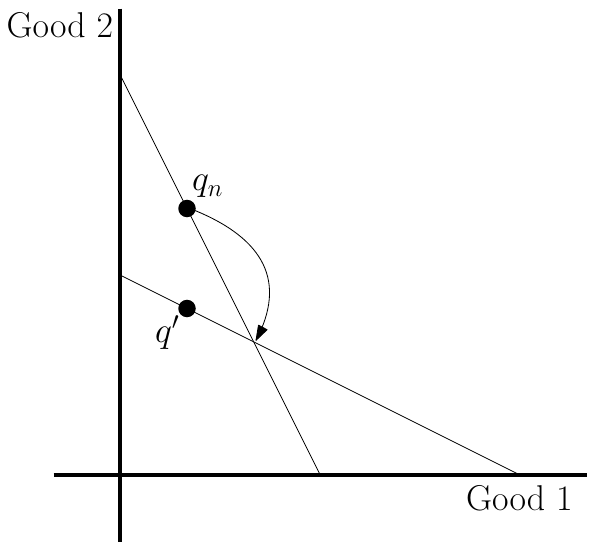}
		\caption{Dataset $D_1$ satisfies GARP.}
		\label{fig:D_1}
	\end{subfigure}
	\hfill
	\begin{subfigure}[b]{0.45\textwidth}
		\centering
		\includegraphics[width=\textwidth]{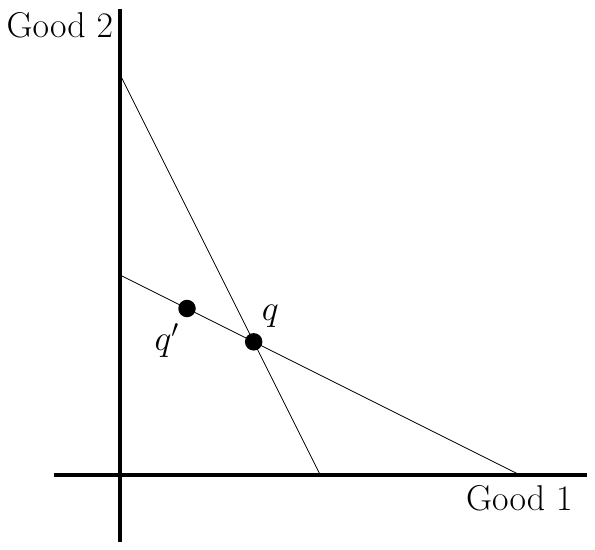}
		\caption{Dataset $\widebar D$ does not satisfy GARP.}
		\label{fig:D}
	\end{subfigure}
	\caption{The dataset on the left, referred to as $D_1$ in Example \ref{example:bad-news}, satisfies GARP. The sequence $\bs{q}_n = (4 - \tfrac{1}{3n}, 4 + \tfrac{2}{3n})$ tends to the location where the two budget lines cross (i.e.\ where the arrow is pointing in the figure to the left). The dataset which contains this limiting bundle, depicted on the right, does not satisfy GARP.}
	\label{fig:bad_limit}
\end{figure}

\begin{example} \label{example:bad-news}
	For each $n$ let $\bs{q}_n = ( 4 - \tfrac{1}{3n}, 4 + \tfrac{2}{3n} )$ and $\bs{q} = \lim_n \bs{q}_n = ( 4,4 )$. Also, let $\bs{q}' = ( 2, 5 )$, $\bs{p} = ( 2,1 )$, and $\bs{p}' = (1,2)$. Assemble these objects into purchase datasets $D_n = (( \bs{q}_n, \bs{p} ), ( \bs{q}', \bs{p}' ) )$ and $\widebar D = ( (\bs{q}, \bs{p}), (\bs{q}', \bs{p}') )$. The datasets $D_1$ and $\widebar D$ are pictured in Figure \ref{fig:bad_limit}. Note that $\bs{p} \cdot \bs{q} = \bs{p} \cdot \bs{q}_n = \bs{p}' \cdot \bs{q}' = 12$, $\bs{p} \cdot \bs{q}' = 9$, and $\bs{p}' \cdot \bs{q}_n = 12 + \tfrac{1}{n}$. Thus, $\bs{q}_n \ P \ \bs{q}'$ but $\bs{q}' \ \cancel{R} \ \bs{q}_n$ and so $D_n$ satisfies GARP. But, $\bs{q} \ P \ \bs{q}'$ and $\bs{q}' \ R \ \bs{q}$ and so $\widebar D$ violates GARP. 
	
	Suppose that $\tau$ is a goodness-of-fit index which is continuous and accurate for the family of well-behaved utility functions. As each $D_n$ satisfies GARP we have $\tau(D_n) = 0$ for each $n$. Since $\bs{q}_n \rightarrow \bs{q}$ and because $\tau$ is continuous we must have $\tau(\widebar D)= \lim_n \tau(D_n) = 0$ which means that $\tau$ is not accurate since $\widebar D$ violates GARP and is not rationalizable by a well-behaved utility function.\footnote{\label{footnote:SARP}The datasets $D_n$ not only satisfy GARP but they also satisfy the strong axiom of revealed preferences (SARP): a dataset satisfies SARP if, for all $t_1, t_2, \ldots, t_K$ we have $\bs{q}^{t_1} \ R \ \bs{q}^{t_2} \ R \ \ldots \ R \ \bs{q}^{t_K} \ R \ \bs{q}^{t_1}$ implies $\bs{q}^{t_1} = \bs{q}^{t_2} = \ldots = \bs{q}^{t_K}$. \citet{matzkin-richter91} show that a dataset satisfies SARP if and only if it can be rationalized by a continuous, increasing, and strictly concave utility function. As the datasets $D_n$ in our example satisfy SARP the example also shows that there are {\em no goodness-of-fit indices which are continuous and accurate for the collection of continuous, increasing, and strictly concave utility functions.}}  
\end{example}

This example exploits the fact that the collection of datasets which obey GARP is not closed, in the sense that there could be a sequence of GARP-obeying datasets $D_n$ with a limit $\bar D$ that violates GARP, and if the sequence of datasets have an index of zero, then continuity will force its limit to also have an index of zero.  The index is then inaccurate in the sense that it fails to identify a dataset violating GARP.  

The dataset $\widebar D$ in Example \ref{example:bad-news} belongs to a class of datasets which any continuous goodness-of-fit index will have difficulty identifying as not rationalizable. To be more precise, we refer to an RP-cycle in a dataset $D$ as {\em weak} if there is at least one link in the cycle (\ref{violateGARP}) that {\em cannot} be replaced by $P$; otherwise, we refer to it as a {\em strong} RP-cycle.  The RP-cycle in dataset $D$ in Example \ref{example:bad-news} is a weak RP-cycle, since $\bs{q} \ P \ \bs{q}'$ and $\bs{q}' \ R \ \bs{q}$ but we cannot replace $R$ with $P$ in the latter relation.    

The next proposition generalizes Example \ref{example:bad-news} by showing that the the issue it highlights holds for all datasets which only contain weak cycles.  

\begin{proposition} \label{prop:badnews2}
	Suppose the purchase dataset $D = (\bs{q}^t, \bs{p}^t)_{t \leq T}$ violates GARP but has only weak RP-cycles. Then there is a sequence of GARP-obeying datasets $D_n=(\bs{q}_n^t, \bs{p})_{t \leq T}$ with $\bs{p}^t \cdot \bs{q}_n^t = \bs{p}^t \cdot \bs{q}^t$ for all $n$ and $t$ and $\lim_{n\to\infty} \bs{q}^t_n=\bs{q}^t$ for all $t$.
	
	Thus, if a goodness-of-fit index $\tau$ for ${\cal U}_{WB}$ is continuous, then $\tau (D)=0$, which means that $\tau$ is not accurate.  It follows that  there are no continuous and accurate goodness-of-fit indices for $\mathcal{U}_{WB}$. 
\end{proposition}  


The proof of this result is in the appendix.  Note that the existence of the sequence of GARP-obeying datasets $D_n$ claimed in this proposition is obvious when there are just two observations (such as in Example \ref{example:bad-news}). The difficulty in establishing this proposition lies in showing that the desired sequence exists no matter how many observations are present in the data. Our proof uses an induction argument on the number of observations.

\subsection{Essential accuracy}

While there are no goodness-of-fit indices which are both continuous and accurate, the news is not all bad, because there are indices (some of which we introduce in Section \ref{sec:popular-indices}) where the inaccurary is confined to datasets with only weak RP-cycles.  To be specific, these indices have the following property:
\begin{equation}\label{strongRP}
\tau (D)>0\:\: \mbox{whenever $D$ contains a strong RP cycle.}
\end{equation}
In other words, if for some dataset $D$ we have $\tau (D)=0$, then either $D$ obeys GARP or $D$ contains only weak RP-cycles; when the latter holds, we know from Proposition \ref{prop:badnews2} that there are datasets close to $D$ that do satisfy GARP (and are therefore rationalizable).  These indices are thus {\em essentially accurate} in the sense we define below.  

\begin{definition} \label{def:weak-accuracy}
	A goodness-of-fit index $\tau$ for $\cal U$ is \emph{essentially accurate} if, for any purchase dataset $D = (\bs{q}^t, \bs{p}^t)_{t \leq T}$ for which $\tau(D) = 0$, there is a sequence of datasets $D_n=(\bs{q}_n^t, \bs{p})_{t \leq T}$ with $\bs{p}^t \cdot \bs{q}_n^t = \bs{p}^t \cdot \bs{q}^t$ for all $n$ and $\lim_{n\to\infty} \bs{q}^t_n=\bs{q}^t$, such that $D_n$ admits a rationalization in $\cal U$. 
\end{definition}

It is obvious that if a dataset contains a strong RP cycle, then any dataset arising from small perturbation of the bundles will also contain a strong RP cycle and therefore cannot be rationalized. The following result summarizes our observations.

\begin{proposition} \label{prop:weak-accuracy2}
	A goodness-of-fit index $\tau$ is essentially accurate for $\mathcal{U}_{WB}$ if and only if it satisfies \eqref{strongRP}.  
\end{proposition}

\subsection{Accuracy versus continuity}

Given that accuracy and continuity are incompatible, we are obliged to choose a goodness-of-fit index with either one property or the other.  We either embrace accuracy and assign {\em all} datasets with GARP violations positive index values or we embrace continuity and lump datasets with only weak RP cycles  along with those which are rationalizable. We assert that how this trade-off is resolved affects more than those datasets with weak RP cycles. 

For instance, consider a new dataset $D^*$ formed by moving the bundle $\bs{q}$ southeast along its budget line in dataset $\widebar D$ depicted in Figure \ref{fig:bad_limit}; then, since we have created a strong RP cycle, a reasonable index $\tau$ should assign $\tau(D^*) \geq \tau(\widebar D)$.  Next consider a dataset $D^{**}$ created by altering $\widebar{D}$ by moving $\bs{q}$ northwest along its budget line. Such a dataset obeys GARP. If $\tau$ is some goodness-of-fit index for $\mathcal{U}_{WB}$ then $\tau (D^{**})=0$ and if $\tau$ is accurate then the gap $\tau (D^{*})-\tau (D^{**})$ will be bounded away from zero, even if $D^{**}$ and $D^*$ are arbitrarily close.  On the other hand, if $\tau $ is continuous, this gap can be made arbitrarily small.  It is clear from these observations, that if there is an empirical analysis on a population of consumers, and many of them have purchasing datasets which are close to $\widebar D$ (whether rationalizable such as $D^{**}$ or not such as $D^*$), then the outcome of the analysis will be sensitive to the index used.\footnote{For example, suppose one would like to use the goodness-of-fit index as a variable to explain wealth or other household characteristics (such as in \citet{choi-kariv14}). Then if the datasets from many households is clustered near $\widebar D$, the analysis could be very sensitive to whether an accurate but dis-continuous or a continuous but inaccurate index was utilized.}


Continuous and essentially accurate indices such as Afriat's index are widely adopted by researchers.  These indices conflate datasets that obey GARP with those that violate GARP but have only weak RP cycles. While we are {\em not} arguing against the use of these indices, it is worth emphasizing that these two types of datasets have important differences.  

Firstly, weak RP cycles {\em are} violations of GARP, which means that there is no well-behaved utility function that rationalizes the data. Consequentially, for any well-behaved $U$ there must be some observation wherein the consumer fails to maximize $U$ or equivalently there is some $t$ wherein the consumer could have obtained utility $U(\bs{q}^t)$ for less money. While it may be true that the analyst cannot form a nontrivial lower bound on the money lost,\footnote{A violation of utility-maximization implies a violation of cost efficiency; see the discussion of the Afriat index in Section \ref{sec:popular-indices} and also the further discussion of Example \ref{example:bad-news} in that section.}  nonetheless there is still a loss.  And if purchasing decisions such as the ones in the dataset $D$ in Example \ref{example:bad-news}  are made by the consumer repeatedly, then these losses add up to a larger amount (for any {\em fixed} utility function). 

This brings us to the second issue which is that there is heterogeneity even among datasets that have only weak RP cycles. When a dataset has many weak RP cycles, it seems reasonable to consider this a more severe GARP violation than a dataset with just one weak RP cycle or even one strong RP cycle, but an essentially accurate (but inaccurate) index is unable to discrminate among such datasets.  (The Houtman-Maks index, explained in Section \ref{sec:popular-indices}, is an example of an accurate but discontinuous index than {\em can} distinguish between datasets with single and multiple weak RP cycles.)   

Lastly, welfare predictions in datasets that violate GARP are less clear cut and this is true even if there are only weak RP cycles.   Observe that, in Example \ref{example:bad-news},  it is uncontroversial to conclude that the consumer prefers $\bs{q}_n$ to $\bs{q}'$, since $\bs{q}_n$ is revealed preferred to $\bs{q}'$ in $D_n$, but the welfare comparison between $\bs{q}'$ and $\bs{q}$ in the dataset $\bar{D}$ is not clear. In fact, that conclusion is sensitive to the index used, even when both indices are continuous and essentially accurate (see Example \ref{example:robust-prefs} in Section \ref{sec:robust-prefs}).

\subsection{Loss Functions and Goodness-of-fit Indices} \label{sec:popular-indices}

In this section we survey five popular goodness-of-fit indices: the Afriat index (from \citet{afriat73}); the Varian index (from \citet{varian90}); the Swaps index (from \citet{apesteguia-ballester15}); the Houtman-Maks index (from \citet{houtman1985}); and the non-linear least squares (LS) index. In Proposition \ref{prop:popular-indices} below we show that the Afriat, Varian, Swaps, and LS indices are continuous but inaccurate while the Houtman-Maks index is accurate but discontinuous. To proceed we first define the concept of a loss function and show how a loss function generates a goodness-of-fit index; each of the popular indices we just listed is generated by some loss function. 

A \emph{loss function}, denoted $Q(U;D)$, is a map from utility functions and purchase datasets to the non-negative real numbers which satisfies $Q(U;D) = 0$ in the event that $U$ rationalizes $D$. Intuitively, lower values of $Q(U; D)$ indicate that $U$ provides a better description of the data. Let $\mathcal{U}$ be some collection of utility functions. The loss function $Q$ \emph{generates} a goodness-of-fit index for $\mathcal{U}$ via
\begin{equation} \label{eq:tau-gen}
	\tau(D) = \inf_{U \in \mathcal{U}} \ Q(U;D).
\end{equation}

\paragraph{The Afriat Index.} A purchase dataset $D = (\bs{q}^t, \bs{p}^t)_{t \leq T}$ is $e$-rationalized (for some $e\in[0,1]$) by utility function $U$ if, for all $t$, we have $U(\bs{q}^t) \geq U(\bs{q})$ for all $\bs{q} \in B(\bs{p}^t, e \bs{p}^t \cdot \bs{q}^t)$ (where $B$ is defined by (\ref{eq:budget-set})).


The \emph{Afriat loss function} is $$A(U;D) = 1 - \sup 
\Big\{ e \in [0,1]: D \text{ is } e\text{-rationalized by } U \Big\}$$ The \emph{Afriat index} for $\mathcal{U}$ is $A(D) = \inf_{U \in \mathcal{U}} A(U;D)$.


Note that if $U$ rationalizes $D$ then $A(U;D)=0$ (since $D$ is 1-rationalized by $U$), so a straightforward interpretation of $A(U;D)$ is that it is a measure of the extent to which bundles in the true budget set $B(\bs{p}^t, \bs{p}^t \cdot \bs{q}^t)$ have to be omitted before $\bs{q}^t$ is optimal among the remaining bundles.  There is also another interpretation of $A(U;D)$ when $\cal U$ is contained in the class of continuous functions (so $U$ is continuous).  In that case, it is straightforward to check that if $A(U;D)=1-e'$, then there must be some observation $\underline t\leq T$ such that 
$$U(\bs{q}^{\underline t})=\max\{U(\bs{q}):\bs{q}\in B(\bs{p}^{\underline t}, e'\bs{p}^{\underline t} \cdot \bs{q}^{\underline t})\}.$$   
This means that there is at least one observation (and possibly more) where the consumer is cost inefficient in the sense that he could have achieved the same utility level with less money ($e'\bs{p}^{\underline t} \cdot \bs{q}^{\underline t}$ rather than $\bs{p}^{\underline t} \cdot \bs{q}^{\underline t}$).    It is for this reason that in \citet{afriat73}, the index $A(D)$ is referred to as the {\em critical cost efficiency index}.


\paragraph{The Varian Index.} A purchase dataset $D = (\bs{q}^t, \bs{p}^t)_{t \leq T}$ is $\bs{e}$-rationalized (for $\bs{e} = (e_1,e_2, \ldots, e_T) \in [0,1]^T$) by utility function $U$ if, for all $t$, we have $U(\bs{q}^t) \geq U(\bs{q})$ for all $\bs{q} \in B(\bs{p}^t, e_t \bs{p}^t \cdot \bs{q}^t)$. Let $f: [0,1]^T \rightarrow \mathbb{R}_+$ be a continuous and decreasing function\footnote{$f$ is decreasing if $\bs{e} > \bs{e}'$ implies $f(\bs{e}) < f(\bs{e}')$.} which satisfies $f( e_1, e_2, \ldots, e_T) = 0$ if and only if $e_1 = e_2 = \ldots = e_T = 1$. The \emph{Varian loss function} is $$V(U;D) = \inf \Big\{ f(\bs{e}): D \text{ is } \bs{e}\text{-rationalized by } U \Big\}$$ The \emph{Varian index} for $\mathcal{U}$ is $V(D) = \inf_{U \in \mathcal{U}} V(U;D)$.\footnote{Of course the values taken by the Varian index depend on one's choice of $f$. All the results we establish hold independent of the $f$ used and so we consider the Varian index for a fixed $f$.}

The concept of an $\bs{e}= (e_1, e_2, \ldots, e_T)$-rationalization has similar logic as the Afriat index's $e$-rationalization. Intuitively, the number $e_t$ represents the level of efficiency exhibited by the consumer in period $t$. The function $f$ serves to aggregate the period specific efficiency numbers into one number representing the aggregate level of waste exhibited by the consumer.

Notice that if we choose $f(\cdot) = 1-\min(\cdot)$ then the Varian index becomes the Afriat index; however, this function $f$ is only weakly decreasing rather than decreasing, and thus the Afriat index is {\em not} a special case of the Varian index as we have defined it.  The requirement that $f$ is decreasing is needed for Proposition \ref{prop:popular-loss-functions} and Proposition 1 in the online appendix to hold.

\paragraph{The Swaps Index.} Let $\mu$ be the Lebesgue measure. The \emph{Swaps loss function} is
\begin{equation*} \label{eq:swaps-def}
	S(U;D) = \sum_{t=1}^T \mu\Big\{ \bs{q} \in B(\bs{p}^t, \bs{p}^t \cdot \bs{q}^t): U(\bs{q}) > U(\bs{q}^t) \Big\}
\end{equation*}
The \emph{Swaps index} for $\mathcal{U}$ is $S(D) = \inf_{U \in \mathcal{U}} S(U;D)$.

For each $t$, the Swaps loss function measures (using the Lebesgue measure) the affordable part of the strict upper contour set of $\bs{q}^t$ and then adds up these measures across $t$. As each bundle in this upper contour set is a bundle that should have been chosen over $\bs{q}^t$ the size of this set represents the degree of mis-optimization exhibited by the consumer. 

\paragraph{Non-Linear Least Squares.} For each purchase dataset $D = (\bs{q}^t, \bs{p}^t)_{t \leq T}$ and each $t \leq T$ let $\bs{W}_D^t$ be an $L \times L$ positive definite matrix whose entries are continuous as a function of the data $D$. The \emph{non-linear least squares (LS) loss function} for dataset $D = (\bs{q}^t, \bs{p}^t)_{t\leq T}$ is
\begin{IEEEeqnarray*}{rl}
	\ls( U; D ) = \inf \bigg\{ \sum_{t=1}^T ( \bs{q}^t - \tbs{q}^t )' \bs{W}_D^t (\bs{q}^t - \tbs{q}^t) : (\tbs{q}^t, \bs{p}^t)_{t \leq T} & \text{ is rationalized by } U \\ & \text{ and } \tbs{q}^t \in B(\bs{p}^t, \bs{p}^t \cdot \bs{q}^t) \bigg\}
\end{IEEEeqnarray*}
The \emph{LS index} for $\mathcal{U}$ is $\ls(D) = \inf_{U \in \mathcal{U}} \ls(U;D)$. 

$\ls(U;D)$ reports the distance (where distance is measured using the weighting matrix $\bs{W}_D^t$) between the purchase dataset $D$ and the collection of datasets which are rationalized by the utility function $U$ and which have the same budget sets observed in $D$. Of course, if the dataset is rationalized by $U$ then this distance is $0$. Note that our definition of the LS loss function includes the case where distance is measured in terms of budget shares. To see this, let $\bs{W}_D^t$ be the diagonal matrix with $(p_{\ell}^t)^2 / (\bs{p}^t \cdot \bs{q}^t )^2$ on diagonal $\ell$. Now, $( \bs{q}^t - \tbs{q}^t )' \bs{W}_D^t (\bs{q}^t - \tbs{q}^t)$ measures the distance between the budget share vector for $(\bs{q}^t, \bs{p}^t)$ and the budget share vector for $(\tbs{q}^t, \bs{p}^t)$.

\paragraph{The Houtman-Maks Index.} The \emph{Houtman-Maks loss function} is
\begin{equation*}
	H\left( U; (\bs{q}^t,  \bs{p}^t )_{t\leq T} \right) = \inf \left\{ | \mathcal{T}| \subseteq \{1,2,\ldots, T\}: ( \bs{q}^t, \bs{p}^t )_{t \notin \mathcal{T}} \text{ is rationalized by } U \right\}
\end{equation*}
where, by convention, the empty set is rationalized by every $U$. The Houtman-Maks index for $\mathcal{U}$ is $H(D) = \inf_{U \in \mathcal{U}} H(U;D)$. The Houtman-Maks loss function reports the number of observations which must be thrown out in order for the purchase data to be rationalized by the utility function. 

We know from Proposition \ref{prop:badnews2} that the aforementioned indices cannot be both continuous and accurate for the class of well-behaved utility functions. The following proposition provides the details.

\begin{proposition} \label{prop:popular-indices}
	The Afriat, Varian, Swaps, and LS indices are continuous but inaccurate for ${\cal U}_{WB}$. The Houtman-Maks index is accurate for ${\cal U}_{WB}$ but discontinuous.\footnote{Lest one think there is something special about Houtman-Maks, we note that the money pump index from \citet{echenique11} is also accurate and discontinuous. This index is based on the idea that to every GARP violation there exists an opportunity for an arbitrageur to extract money from the consumer and its value is obtained by summing up the money which can be extracted in each period. For instance, the dataset $\widebar D$ from Example \ref{example:bad-news} violates GARP but has only one weak RP-cycle; nonetheless, the value of the money pump, given by  
$\bs{p} \cdot (\bs{q}-\bs{q}')+\bs{p}'\cdot (\bs{q}'-\bs{q})$, is strictly positive.}  
\end{proposition}

To give an illustration of this proposition we introduce a well-known result on the Afriat index. Let $D = (\bs{q}^t, \bs{p}^t)_{t \leq T}$ be a purchase dataset and let $e \in [0,1]$. Let $\bs{q}^t \ R_e \ \bs{q}$ mean $e \bs{p}^t \cdot \bs{q}^t \geq \bs{p}^t \cdot \bs{q}$ and let $\bs{q}^t \ P_e \ \bs{q}$ mean $e \bs{p}^t \cdot \bs{q}^t > \bs{p}^t \cdot \bs{q}$. We refer to $R_e$ and $P_e$ as the $e$ direct and strict direct revealed preference relations for $D$. We say that $D$ satisfies $e$-GARP if, for all $t_1, t_2, \ldots, t_K$
\begin{equation*}
	\bs{q}^{t_1} \ R_e \ \bs{q}^{t_2} \ R_e \ \ldots \ R_e \ \bs{q}^{t_K} \qquad \implies \qquad \bs{q}^{t_K} \ \cancel{P_e} \ \bs{q}^{t_1}
\end{equation*}
Recall that $A$ is the Afriat index. The following is a consequence of Theorem 1 in \citet{halevypersitzzrill18}.
\begin{proposition} \label{prop:AEI}
	For any purchase dataset $D$
	\begin{equation} \label{eq:VEI}
		A(D) = 1 - \sup \Big\{ e \in [0,1]: D \text{ satisfies } e\text{-GARP} \Big\} 
	\end{equation}
\end{proposition}
A similar result could be stated for the Varian index but this proposition suffices for our purposes. We now return to Example \ref{example:bad-news} to elaborate on Proposition \ref{prop:popular-indices}. 

\addtocounter{example}{-1}
\begin{example}[continued]
	Let $\widebar D = ( (\bs{q}, \bs{p}), (\bs{q}', \bs{p}') )$ be the dataset introduced in Example \ref{example:bad-news} and pictured in Figure \ref{fig:D}. Clearly, $\bs{q}' \ \cancel{R_e} \ \bs{q}$ for all $e < 1$ and, consequently, $\widebar D$ satisfies $e$-GARP for all $e < 1$. From Proposition \ref{prop:AEI} we see that $A(\widebar  D) = 0$. Since $D$ violates GARP, we may conclude that $A$ is inaccurate.

Notice that the fact that $\widebar D$ violates GARP also means that $A(U; \widebar D)>0$ for any $U\in{\cal U}_{WB}$.\footnote{We consider two cases (a) $U( \bs{q}' ) \geq U( \bs{q} )$ and (b) $U( \bs{q} ) > U(\bs{q}')$. Under case (a) we have $A(U;\widebar D) \geq 1- \bs{p} \cdot \bs{q}' / \bs{p} \cdot \bs{q} > 0$. Under case (b) because $U$ is well-behaved there exists some bundle $\tbs{q}$ so that $\bs{q} > \tbs{q}$ and $U(\tbs{q}) \geq U(\bs{q}')$. Consequently, $A(U;\widebar D) \geq 1- \bs{p}' \cdot \tbs{q} / \bs{p}' \cdot \bs{q}' > 0$.}   Therefore, there must be a sequence of $U_n\in {\cal U}_{WB}$ such that $A(U_n;\widebar D)>0$ but $\lim_{n\to\infty}A(U_n;\widebar D)= 0$, i.e., there is a sequence of utility functions for which the agent's cost inefficiency is strictly positive but arbitrarily small.   

With the Houtman-Maks index, it is easy to see that $H(\widebar D) = 1$, so the index does not assign $\widebar D$ the inaccurate value of $0$. But note that $H(D_n) = 0$ for all $n$ where $D_n$ is the dataset defined in Example \ref{example:bad-news} (dataset $D_1$ is depicted in Figure \ref{fig:D_1}). Thus, $H(D_n) \rightarrow 0 \neq H(\widebar D)$ and so $H$ is not continuous as $D_n \rightarrow \widebar D$.
\end{example}

While the Afriat, Varian, Swaps and LS indices are inaccurate for the class of well-behaved utility functions, these indices are all essentially accurate in the sense of Definition \ref{def:weak-accuracy}.  

\begin{proposition}  \label{prop:essacc-indices}
	The Afriat, Varian, Swaps, and LS indices are essentially accurate for ${\cal U}_{WB}$.
\end{proposition}

Recall that an essentially accurate index $\tau$ is such that if $\tau (D)=0$ then $D$ must either be rationalizable or be the limit of rationalizable datasets. Given Proposition \ref{prop:badnews2}, to prove Proposition \ref{prop:essacc-indices} we need only show that those indices take on a strictly positive value whenever a dataset $D$ has a strong GARP violation.  We show this in the appendix.




\section{Loss functions to estimate utility} \label{sec:bad-news2}

The following definition gives two desirable properties for a loss function.   

\begin{definition} \label{def:objective-functions}
	Let $\mathcal{U}$ be a collection of utility functions and let $Q$ be a loss function. We say that $Q$ is \emph{accurate for $\mathcal{U}$} if, for each $U \in \mathcal{U}$ we have $Q( U; D ) = 0$ if and only if $D$ is rationalized by $U$. We say $Q$ is \emph{minimizable for $\mathcal{U}$} if, for all datasets $D$, the set  
	\begin{equation*} \label{eq:extremum-est}
		\underset{ U \in \mathcal{U}}{\argmin} \ Q(U;D)
	\end{equation*}
	is nonempty.  
\end{definition}

Accuracy says that the loss function attains a perfect score of $0$ if and only if the data is perfectly explained by the utility function.  Minimizability requires that best-fitting utility functions exist. This is a desirable property to have since, if $\argmin_{U\in\mathcal{U}}Q(U;D)$ is nonempty, then it is natural to consider its elements as estimators of the consumer's utility function. The following result relates the properties of a loss function just introduced with the accuracy of its associated index.   

\begin{proposition} \label{prop:accurate}
	Let $\mathcal{U}$ be a collection of utility functions, let $Q$ be a loss function, and let $\tau$ be defined by \eqref{eq:tau-gen}. If $Q$ is accurate and minimizable for $\mathcal{U}$ then $\tau$ is accurate for $\mathcal{U}$.
\end{proposition}
\begin{proof}
	Suppose $Q$ is accurate and minimizable. Suppose $D$ is rationalized by some $U \in \mathcal{U}$. Then, $Q(U, D) = 0$ and thus $\tau(D) = 0$. On the other hand, suppose $D$ is not rationalized by any $U \in \mathcal{U}$. As the infimum in \eqref{eq:tau-gen} is attained there exists some $U \in \mathcal{U}$ so that $Q(U;D) = \tau(D)$. As $U$ does not rationalize $D$ we have $0 < Q(U;D) = \tau(D)$ and so $\tau$ is accurate. 
\end{proof}

An immediate consequence of Proposition \ref{prop:accurate} is the following impossibility result on loss functions.   

\begin{proposition} \label{prop:bad-news2}
	There does not exist a loss function which is accurate and minimizable for ${\cal U}_{WB}$ and which generates a continuous goodness-of-fit index via \eqref{eq:tau-gen}.
\end{proposition}
\begin{proof}
By Proposition \ref{prop:accurate}, a loss function which is accurate and minimizable for some $\mathcal{U}$ generates an accurate goodness-of-fit index via \eqref{eq:tau-gen}. However, Proposition \ref{prop:badnews2} tells us that an accurate and continuous goodness-of-fit index does not exist for ${\cal U}_{WB}$.
\end{proof}


\subsection{Popular Loss Functions} \label{sec:popular-loss-functions}

From Propositions \ref{prop:popular-indices} and \ref{prop:bad-news2} we know that the Afriat, Varian, Swaps, and LS loss functions cannot be accurate and minimizable for well-behaved utility functions.  Recall that we have already shown that the Afriat loss function is not minimizable: in our discussion of Example \ref{example:bad-news} in Section \ref{sec:popular-indices}, we have shown that, for the dataset $D$ in that example, $A(D)=0$ but $A(U;D)>0$ for all $U\in{\cal U}_{WB}$.  In fact, that example is typical; there is a precise sense in which the Afriat loss function is not minimizable for all but a rare set of datasets.

The next result gives the precise ways in which each loss function falls short.   

\begin{proposition} \label{prop:popular-loss-functions}
	The Afriat, Varian, Swaps, and LS loss functions are accurate but not minimizable for $\mathcal{U}_{WB}$.  More specifically, the following holds.
	\begin{enumerate}
		\item Let $Q$ be either the Varian, Swaps, or LS loss function. If $D$ does not satisfy GARP then $\argmin_{U \in \mathcal{U}_{WB}} Q(U; D)$ is empty.
		\item Let $T \in \mathbb{N}$ and let $\mathcal{D}$ be the collection of $T$ observation purchase datasets which do not satisfy GARP. We regard $\cal D$ as a subset of $( (\mathbb{R}_+^L \backslash \bs{0}) \times \mathbb{R}_{++}^L )^T$. For the Afriat loss function $A$, the set $$\{ D \in \mathcal{D}: \argmin_{U \in \mathcal{U}_{WB}} A(U; D) \neq \emptyset \}$$ 
is rare, i.e., its closure has an empty interior in $\cal D$.     
	\end{enumerate}
\end{proposition}


The proof of Proposition  \ref{prop:popular-loss-functions} is in the appendix. It is straightforward to show that the listed loss functions are accurate.   The proof that $\argmin_{U \in \mathcal{U}_{WB}} Q(U; D)$ is empty when $Q$ is the Varian or Swaps loss functions makes use of Theorem 2 in \citet{nishimura-oK-quah17} (stated, for convenience, as Lemma \ref{lemma:NOQ} in the appendix). We use this result to provide a formula for $\inf_{U \in \mathcal{U}_{WB}} Q(U;D)$; based on this formula we can prove that no well-behaved utility function $\ti{U}$ satisfies $Q(\ti{U};D) = \inf_U Q(U;D)$.



Proposition \ref{prop:popular-loss-functions} states that $\argmin_{U \in \mathcal{U}_{WB}} Q(U; D)$ does not typically exist for the Afriat, Varian, Swaps, or LS loss functions. This suggests that there could be complications to drawing conclusions about a consumer's preference, since we cannot simply rely on a best-fitting utility function. In the next subsection we propose a method for making welfare comparisons that works even when best-fitting utility functions do not exist. 

\subsection{Welfare when there are no best-fitting utility functions} \label{sec:robust-prefs}

Given a purchase dataset $D = (\bs{q}^t, \bs{p}^t)_{t \leq T}$ and two consumption bundles $\tbs{q}$ and $\tbs{q}'$ (not necessarily among the bundles observed in $D$), when could we conclude that the consumer prefers $\tbs{q}$ to $\tbs{q}'$?  Suppose that $D$ can rationalized by some utility function in the family $\cal U$; then a natural criterion for concluding that $\tbs{q}$ is preferred to $\tbs{q}'$ is to require $U(\tbs{q})\geq U(\tbs{q}')$ for every utility function $U \in \cal U$ that rationalizes $D$.\footnote{This is effectively the criterion proposed in \citet{varian82}, where $\cal U$ is the family of well-behaved and concave utility functions. The approach is also used in \citet{chambers25} (see their Theorem 1).}  When $D$ cannot be (exactly) rationalized by an element in $\cal U$, one could extend the criterion to require that $\tbs{q}$ be preferred to $\tbs{q}'$ for every best-fitting utility function (with respect to some loss function); however, as we have pointed out, best-fitting utility functions may be non-existent for a generic dataset. To overcome this difficulty we propose a criterion which requires that $\tbs{q}$ be preferred to $\tbs{q}'$ for all utility functions that are, in a sense, sufficiently close to minimizing the loss function. Formally, we propose the following binary relation for making welfare statements.    

\begin{definition} \label{def:robust-pref}
	Let $\mathcal{U}$ be a collection of utility functions, $D = (\bs{q}^t, \bs{p}^t)_{t \leq T}$ be some purchase dataset, and $Q$ be a loss function. We say that bundle $\tbs{q}$ is $Q$-\emph{robustly preferred} to $\tbs{q}'$, and we write $\tbs{q} \ \mathcal{R}_{Q} \ \tbs{q}'$, if there exists a utility function $\ti{U} \in \mathcal{U}$ so that $U(\tbs{q}) \geq U( \tbs{q}' )$ for all $U \in \mathcal{U}$ satisfying $Q(U;D) \leq Q(\ti{U}; D)$. 

	We refer to ${\cal R}_Q$ as the {\em robust preference} relation.  
\end{definition}

In other words, $\tbs{q}$ is $Q$-robustly preferred to $\tbs{q}'$ if there exists some $\ti{U} \in \mathcal{U}$ so that for any $U \in \mathcal{U}$ which provides a better fit (according to $Q$) than $\ti{U}$ it happens that $U$ ranks $\tbs{q}$ above $\tbs{q}'$. Notice that in the event that 
$\argmin_{U \in \mathcal{U}} Q(U; D)$ is non-empty and best-fitting utility functions exist, $\ti{U}$ can be chosen to be any best fitting utility function and we obtain the natural criterion that $\tbs{q} \ \mathcal{R}_{Q} \ \tbs{q}'$ if and only if $U(\tbs{q}) \geq U( \tbs{q}' )$ for every best-fitting $U \in \mathcal{U}$. On the other hand, if  $\tbs{q}$ is {\em not} $Q$-robustly preferred to $\tbs{q}'$, then there are utility functions that come arbitrarily close to minimizing the loss function, and with the utility of $\tbs{q}'$ being strictly higher than that of $\tbs{q}$; formally, there is a sequence of utility function $U_n\in\cal U$ such that $\lim_{n\to\infty} Q(U_n;D)= \inf_{U\in{\cal U}}Q(U;D)$ and $U_n(\tbs{q}')>U_n(\tbs{q})$.   

The following result establishes that the robust preference relation satisfies certain natural properties and also inherits certain properties from the collection $\mathcal{U}$.

\begin{proposition} \label{prop:uniform-pref}
	Let $\mathcal{U}$ be a collection of utility functions, $D = (\bs{q}^t, \bs{p}^t)_{t \leq T}$ be some purchase dataset, and $Q$ be a loss function. Then
	\begin{enumerate}
		\item $\mathcal{R}_Q$ is transitive and reflexive.
		\item If, for some $\bs{q}, \tbs{q}$ we have $U(\bs{q}) \geq U(\tbs{q})$ for all $U \in \mathcal{U}$ then we also have $\bs{q} \ \mathcal{R}_{Q} \ \tbs{q}$.
		\item If, for some $\bs{q}, \tbs{q}, \bs{q}', \tbs{q}'$ we have
		\begin{equation*}
			U(\bs{q}) \geq U(\tbs{q}) \qquad \implies \qquad U(\bs{q}') \geq U(\tbs{q}'), \qquad \forall U \in \mathcal{U}
		\end{equation*}
		then
		\begin{equation*}
			\bs{q} \ \mathcal{R}_{Q} \ \tbs{q} \qquad \implies \qquad \bs{q}' \ \mathcal{R}_{Q} \ \tbs{q}'
		\end{equation*}
	\end{enumerate}
\end{proposition}

Here are some obvious consequences of this proposition. If $\mathcal{U} = \mathcal{U}_{WB}$ then, by statement 2, we see that $\mathcal{R}_{Q}$ is increasing. If each $\bs{q}$ represents a vector of contingent consumption over $L$ states and each $U \in \mathcal{U}$ is increasing in the sense of first order stochastic dominance then, again by statement 2, we see that $\mathcal{R}_{Q}$ is also increasing in the sense of first order stochastic dominance. If $\mathcal{U}$ is a collection of homothetic functions then statement 3 tells us that $\bs{q} \ \mathcal{R}_{Q} \ \tbs{q}$ implies $k \bs{q} \ \mathcal{R}_{Q} \ k\tbs{q}$ for all $k > 0$. If $\mathcal{U}$ is a collection of quasi-concave functions then, again by statement 3, $\bs{q} \ \mathcal{R}_{Q} \ \tbs{q}$ implies $(k\bs{q} + (1-k) \tbs{q}) \ \mathcal{R}_Q \ \tbs{q}$ for all $k \in (0,1)$.

In the case of the Afriat loss function and the class of well-behaved utility functions, it is possible characterize the robust preference relations in terms of the $e$-direct and strict direct revealed preference relations (denoted $R_e$ and $P_e$) which we introduced in Section \ref{sec:popular-indices}.\footnote{In the online appendix we provide an analogous result for calculating the robust preference relation in the case of the Varian loss function, with ${\cal U}={\cal U}_{WB}$ and we also discuss how to calculate the robust preference relation for the Houtman-Maks loss function. It turns out that both of these robust preference relations can be computed by solving a mixed integer linear programming problem. We put this into practice in Section \ref{sec:empirical}.} First, we need some new definitions. For two binary relations $\succeq$ and $\succeq'$ on $\rlp$ we write $\bs{q} \ (\succeq \cup \succeq') \ \bs{q}'$ to mean either $\bs{q} \succeq \bs{q}'$ or $\bs{q} \succeq' \bs{q}'$. Further, we write $\tran(\succeq)$ to denote the transitive closure of $\succeq$. In other words, $\bs{q} \ \tran(\succeq) \ \bs{q}'$ means there exists some $\bs{q}_1, \bs{q}_2, \ldots, \bs{q}_N$ so that $\bs{q} \succeq \bs{q}_1 \succeq \bs{q}_2 \succeq \ldots \succeq \bs{q}_N \succeq \bs{q}'$. 


\begin{proposition} \label{prop:afriat-welfare}
Given a purchase dataset $D = (\bs{q}^t, \bs{p}^t)_{t \leq T}$, let ${\cal R}_A$ be the robust preference relation for ${\cal U}={\cal U}_{WB}$ and let 
$$e^* = \sup\{ e \in [0,1]: D \text{ satisfies } e\text{-GARP} \}.$$ 
\begin{itemize}
\item[{\em (a)}]  If $D$ satisfies $e^*$-GARP then $\tbs{q} \ \mathcal{R}_{A} \ \tbs{q}'$ if and only if $\tbs{q} \ \tran( \geq \cup \ R_{e^*} ) \ \tbs{q}'$. 
\item[{\em (b)}]  If $D$ does not satisfy $e^*$-GARP then $\tbs{q} \ \mathcal{R}_A \ \tbs{q}'$ if and only if $\tbs{q} \ \tran( \geq \cup \ P_{e^*} ) \ \tbs{q}'$.
\end{itemize}
\end{proposition}

Proposition \ref{prop:afriat-welfare} gives us a way of computing whether or not $\tbs{q} \ \mathcal{R}_{A} \ \tbs{q}'$. Indeed, provided we can calculate $e^*$, it is easy to determine, using standard algorithms for calculating the transitive closure of a graph (for instance Warshall's algorithm), whether or not $\tbs{q} \ \tran( \geq \cup \ R_{e^*} ) \ \tbs{q}'$ and $\tbs{q} \ \tran( \geq \cup \ P_{e^*} ) \ \tbs{q}'$ hold. The value $e^*$ can be calculated exactly by searching over a finite (and relatively small, containing no more than $T^2$ elements) set of possible values. Details on calculating $e^*$ are found in Section \ref{sec:e^*-calc} of the appendix.\footnote{It is common in empirical applications of Afriat's index for $e^*$ to be calculated via a binary search over the entire interval $[0,1]$, with $e$-GARP being checked at each value of $e$.  This approach leads to an approximation of $e^*$ rather than its exact value.}

\begin{figure}[h]
	\centering
	\begin{subfigure}[b]{0.48\textwidth}
		\centering
		\includegraphics[width=\textwidth]{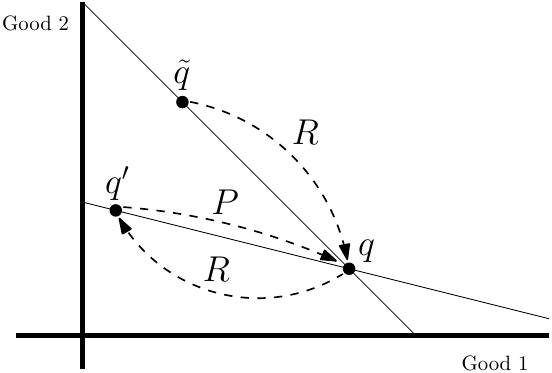}
		\caption{Revealed preference relations.}
		\label{fig:robust-prefs-rp}
	\end{subfigure}
	\hfill
	\begin{subfigure}[b]{0.48\textwidth}
		\centering
		\includegraphics[width=\textwidth]{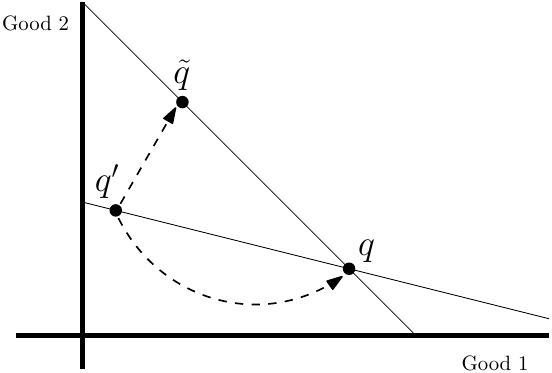}
		\caption{$\mathcal{R}_A$ relations.}
		\label{fig:robust-prefs-A}
	\end{subfigure}
	\vspace{15pt}

	\begin{subfigure}[b]{0.48\textwidth}
		\centering
		\includegraphics[width=\textwidth]{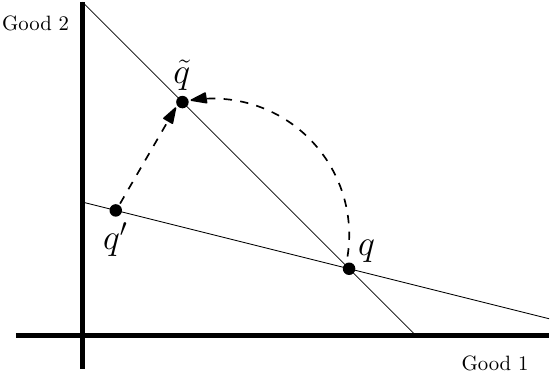}
		\caption{$\mathcal{R}_{\ls}$ relations.}
		\label{fig:robust-prefs-ls}
	\end{subfigure}
	\hfill
	\begin{subfigure}[b]{0.48\textwidth}
		\centering
		\includegraphics[width=\textwidth]{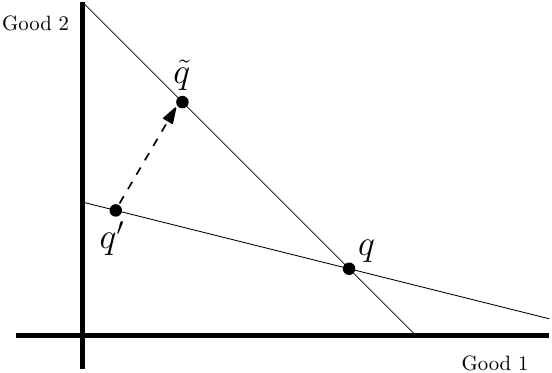}
		\caption{$\mathcal{R}_{HM}$ relations.}
		\label{fig:robust-prefs-HM}
	\end{subfigure}
	\caption{The GARP-violating dataset $\widebar D$. Depicted are various binary relations over the bundles $\bs{q}$, $\bs{q}'$, and $\tbs{q}$.}
	\label{fig:robust-prefs}
\end{figure}

It should not come as a surprise that different loss functions can, at least in principle, lead to different welfare conclusions.  We illustrate this with a dataset $\widebar D$ similar to the one in Example \ref{example:bad-news}.       

\begin{example} \label{example:robust-prefs}
{\em The dataset $\widebar D = ( (\bs{q}, \bs{p}), (\bs{q}', \bs{p}') )$ is depicted in Figure \ref{fig:robust-prefs-rp}.  We have also indicated the revealed preference relations among three bundles, $\bs{q}$, $\bs{q}'$ and $\tbs{q}$; we have $\bs{q} \ P \ \bs{q}'$,  $\bs{q}' \ R \ \bs{q}$, and $\bs{q} \ R \ \tbs{q}$.  We would like to compare how the robust ranking of these bundles change with the Afriat, Least-Squares, and Houtman-Maks loss functions.  Since $\tbs{q} \gg \bs{q}'$ we know (from Proposition \ref{prop:uniform-pref}) that $\tbs{q} \ \mathcal{R}_Q \ \bs{q}'$ for \emph{any} loss function $Q$. A more targeted analysis is needed for the other rankings. 






\vspace{10pt}

\noindent {\bf Afriat Loss Function.}\, As $\bs{q}' \ \cancel{R}_e \ \bs{q}$ for any $e < 1$ we conclude that $1 = \sup \{ e \in [0,1]: {\widebar D} \text{ satisfies } e\text{-GARP} \}$.  Since the dataset does not satisfy GARP we must appeal to part (b) of Proposition \ref{prop:afriat-welfare} to calculate $\mathcal{R}_{A}$, which tells us that $\bs{q} \ \mathcal{R}_A \ \bs{q}'$ since $\bs{q} \ P_1 \ \bs{q}'$. On the other hand, $\bs{q} \ \cancel{P_1} \ \tbs{q}$ and $\bs{q}' \ \cancel{P_1} \ \bs{q}$ and so $\bs{q} \ \cancel{\mathcal{R}_A} \ \tbs{q}$ and $\bs{q} \ \cancel{\mathcal{R}_A} \ \tbs{q}$. The robust relations are depicted in Figure \ref{fig:robust-prefs-A}.  \vspace{10pt}

\noindent {\bf Least Squares Loss Function.}\,  The utility functions that come arbitrarily close to minimizing least squares will rationalize datasets $\ti{D} = ( (\hbs{q}, \bs{p}), (\hbs{q}', \bs{p}') )$ where $\hbs{q}$ is on the budget line and arbitarily close to but northwest of $\bs{q}$ and $\hbs{q}'$ is arbitrarily close to $\bs{q}'$ and thus $\tbs{q}\gg \hbs{q}'$. A utility function $U\in{\cal U}_{WB}$ that rationalizes this perturbed dataset will have $U(\tbs{q})>U(\hbs{q}')\geq U(\bs{q})$.  Therefore, $\tbs{q} \ \mathcal{R}_{\ls} \ \bs{q}$.  On the other hand, it is possible for such a $U$ to satisfy  $U(\bs{q})>U(\bs{q}')$ or $U(\bs{q})< U(\bs{q}')$.  Therefore, $\bs{q} \ \cancel{ \mathcal{R}_{\ls} } \ \bs{q}'$ and $\bs{q} \ \cancel{ \mathcal{R}_{\ls} } \ \tbs{q}'$. The robust relations are depicted in Figure \ref{fig:robust-prefs-ls}; obviously they are different from the ones in Figure \ref{fig:robust-prefs-A}. \vspace{10pt}


\noindent {\bf Houtman-Maks Loss Function.}\, This loss function is minimizable and thus the robust preference relation $\mathcal{R}_{HM}$ satisfies $\bbs{q} \ \mathcal{R}_{H} \ \bbs{q}'$ if and only if $U(\bbs{q}) \geq U(\bbs{q}')$ for all well-behaved $U$ which satisfy $H( U;\widebar D ) = H(\widebar D)$. It is clear that $H(U;\widebar D) = H(\widebar D)$ if and only if $U$ either rationalizes the single observation dataset $( \bs{q}, \bs{p} )$ or the dataset $(\bs{q}', \bs{p}')$. We can find some well-behaved $U$ which rationalizes $(\bs{q}, \bs{p})$ and which satisfies $U(\bs{q}) > U(\tbs{q}) > U(\bs{q}')$ and we can find some well-behaved $\ti{U}$ which rationalizes $(\bs{q}', \bs{p}')$ and which satisfies $\ti{U}(\tbs{q}) > \ti{U}( \bs{q}' ) > \ti{U}( \bs{q} )$ and so $\bs{q} \ \cancel{ \mathcal{R}_{H} } \ \tbs{q}$, $\bs{q} \ \cancel{ \mathcal{R}_{H} } \ \bs{q}'$, and $\bs{q}' \ \cancel{ \mathcal{R}_{H} } \ \bs{q}$. The robust rankings are depicted in Figure \ref{fig:robust-prefs-HM}.} 
\end{example}


\subsection{Loss Functions and Consideration Sets} \label{sec:bounded-rat}


One strand of the literature on bounded rationality emphasizes that consumers may fail to consider all bundles available for purchase and may instead consider only a subset of the affordable bundles called the consideration set (see, for example, \citet{masatlioglu12}). This leads naturally to the problem of measuring the extent to which a consumer is boundedly rational {\em within a particular model of consideration set formation}. In this section we explain how this can be done. 

Let $D=(\bs{q}^t,\bs{p}^t)_{t\leq T}$ be a dataset. The consideration set at observation $t$ is set $B^t$ that contains $\bs{q}^t$ and is contained in the budget set $B(\bs{p}^t , \bs{p}^t \cdot \bs{q}^t)$.  A $T$-tuple of consideration sets $\mathcal{B} = \{ B^t \}_{ t \leq T}$ and a utility function $U$ rationalizes $D$ if $U(\bs{q}^t)\geq U(\bs{q})$ for all $\bs{q}\in B^t$.   


A given model of consideration sets, let us call it $\cal M$, determines $\mathscr{B}_{{\cal M}}(D)$, the collection of $T$-tuples of consideration sets that are permitted by this model.   We assume that full rationality is included as a special case of the model, i.e, the T-tuple ${\cal B}^*=(B(\bs{p}^t , \bs{p}^t \cdot \bs{q}^t))_{t\leq T}$ is always in $\mathscr{B}_{{\cal M}}(D)$.  We say that a utility function $U$ {\em rationalizes $D$ according to the model $\cal M$} if there is $\mathcal{B}\in\mathscr{B}_{{\cal M}}(D)$ such that $\cal B$ and $U$ rationalizes $D$.  Let $\mathscr{B}_{{\cal M}} (U;D)$ denote the subset of $T$-tuples in $\mathscr{B}_{\cal M}(D)$ that, together with $U$, rationalizes $D$.   

We say that $\mathcal{B} = \{ B^t \}_{t \leq T} \in \mathscr{B}_{\cal M}(D)$ is \emph{richer} than $\ti{\mathcal{B}} = \{ \ti{B}^t \}_{t \leq T} \in \mathscr{B}_{\cal M}(D)$ if $\ti{B}^t \subseteq B^t$ for all $t$. It is evident that if $U$ rationalizes $D$ with the richer $\mathcal{B}$ then it can also be rationalized with $U$ using the less rich $\mathcal{B}'$. A \emph{penalty function $\kappa(\cdot; D)$} is a map from $\mathscr{B}_{\cal M}(D)$ to $\mathbb{R}_{+}$ which satisfies the following two properties:
\begin{itemize}
\item[(i)]\, $\kappa( \mathcal{B}; D )=0$ if ${\cal B}={\cal B}^*$ and 
\item[(ii)]\,  $\kappa( \mathcal{B}; D ) \leq \kappa( \ti{\mathcal{B}}; D )$  if  $\mathcal{B}$ is richer than $\ti{\mathcal{B}}$.
\end{itemize}
In other words, there is no penalty if consideration sets coincide with the true budget sets and the penalty increases when consideration sets are smaller across all observations.  

The penalty function $\kappa$ and model $\cal M$ generates a loss function $Q_{\cal M}(U;D)$ given by 
\begin{equation*}
	Q_{\cal M}(U;D) = \inf_{ \mathcal{B} \in \mathscr{B}_{\cal M}(U;D) } \ \kappa(\mathcal{B}; D)
\end{equation*}
i.e., it is the smallest departure from full rationality (as measured by $\kappa$) that can explain the dataset $D$, using the utility function $U$ and the consideration sets permitted by the model $\cal M$.  Obviously, from this loss function, we can define the goodness-of-fit index 
$$\tau_{\cal M}(D)=\inf_{U\in{\cal U}}Q_{\cal M}(U;D),$$ 
which measures the departure from rationality that is required to explain $D$, within the $\cal M$-model of consideration sets.  It almost goes without saying that the general results developed in the previous sections on loss functions and goodness-of-fit indices apply to $Q_{\cal M}$ and $\tau_{\cal M}$.     

Finally, as an illustration, we note that we can interpret the Afriat loss function as arising from a particular model $\cal A$ of consideration sets.\footnote{It is not difficult to verify that the Varian, Swaps, and Houtman-Maks loss functions can also be interpreted as arising from specific models of consideration sets.}  In that case, ${\cal B}\in \mathscr{B}_{\cal A}(D)$ if there is an $e\in [0,1]$ such that 
$$B^t=\{\bs{q}^t\}\cup  B(\bs{p}^t,e\bs{p}^t\cdot\bs{q}^t)$$  
and define $\kappa ({\cal B}; D)=1-e$. Then it is clear that $Q_{\cal A}(U;D)=A(U;D)$. (Relatedly, \citet{dziewulksi20} shows that the Afriat index can be interpreted as a measure of limited perception.)

In the online appendix, we show how the approach in this section can be generalized to incorporate other models of bounded rationality (such as two-stage choice) into loss functions. 


\section{Other collections of utility functions} \label{sec:other}

So far in the paper we have mostly focused on well-behaved utility functions and we have shown that, for this class of utility functions, loss functions and their corresponding goodness-of-fit indices must necessarily be deficient in one aspect or another. The key reason behind this phenomenon is that {\em the collection of datasets which can be rationalized by well-behaved utility functions do not form a closed set.} This is clear from Example \ref{example:bad-news}, where a sequence of rationalizable datasets has a limit which does not admit a rationalization. 

There are some classes of utility functions $\cal U$ where the datasets which admit rationalization by a member of $\cal U$ {\em do} form a closed set; in these cases (and only in these cases), there is a continuous and accurate goodness-of-fit index generated by a minimizable and accurate loss function. The next result states this precisely.  

\begin{proposition} \label{prop:good-news}
	Let $\mathcal{U}$ be a non-empty collection of continuous utility functions and let $\mathcal{D}_{\mathcal{U}}$ denote the collection of purchase datasets which can be rationalized by some $U \in \mathcal{U}$. The following are equivalent.
	\begin{enumerate}
		\item For each $T$, the collection of $T$-observation purchase datasets $\mathcal{D}_{\mathcal{U}}^T = \{ (\bs{q}^t, \bs{p}^t)_{t \leq T} \in \mathcal{D}_{\mathcal{U}} \}$ is closed when considered as a subset of $( (\rlp \backslash \{ \bs{0} \}) \times \mathbb{R}_{++}^L )^T$.\footnote{More specifically, we mean that $\mathcal{D}_{\mathcal{U}}^T$ is closed in the relative topology of $( (\rlp \backslash \{ \bs{0} \}) \times \mathbb{R}_{++}^L )^T$.}
		\item There exists a goodness-of-fit index $\tau$ which is continuous and accurate for $\mathcal{U}$.
		\item There exists a loss function $Q$ which is accurate and minimizable for $\mathcal{U}$ and whose generated goodness-of-fit index (defined by \eqref{eq:tau-gen}) is continuous and accurate for $\mathcal{U}$.
	\end{enumerate}
\end{proposition}
\begin{proof}
	In what follows we treat purchase datasets with $T$ observations as elements of $(\mathbb{R}^L \times \mathbb{R}^L)^T$ and thus, for datasets $D = (\bs{q}^t, \bs{p}^t)_{t \leq T}$ and $\ti{D} = (\tbs{q}^t, \tbs{p}^t)_{t \leq T}$ we have 
$$|| D - \ti{D} ||^2 = \sum_{t=1}^T || \bs{q}^t - \tbs{q}^t ||^2 + || \bs{p}^t - \tbs{p}^t ||^2 $$ 
where $|| \cdot ||$ denotes the usual Euclidean norm. 
	
Clearly 3.\ implies 2. We shall show 2.\ implies 1.\ implies 3. So, suppose there exists a goodness-of-fit indices $\tau$ which is continuous and accurate for $\mathcal{U}$. Let $D_n$ be a sequence in $\mathcal{D}_{\mathcal{U}}^T$ which converges to a purchase dataset $D$. We have $\tau(D_n) \rightarrow \tau(D) = 0$ as $\tau$ is continuous and accurate and $D_n \in \mathcal{D}_{\mathcal{U}}^T$. Now, as $\tau$ is accurate we have $D \in \mathcal{D}_{\mathcal{U}}^T$ and so indeed $\mathcal{D}_{\mathcal{U}}^T$ is closed.
	
	Next, we show that item 1.\ implies 3. So, suppose $\mathcal{D}_{\mathcal{U}}^T$ is closed. Let $\mathcal{D}_U^T$ be the collection of datasets which have $T$ observations and which are rationalized by $U$. Let 
$$Q(U;D) = \inf \{ || D - \ti{D} ||: \ti{D} \in \mathcal{D}_U^T \}.$$ 
As $U$ is continuous we see that $\mathcal{D}_U^T$ is closed (in particular, the demand correspondence generated by $U$ is upper hemicontinuous and thus has a closed graph) and thus $Q(U;D) = 0$ if and only if $D$ is rationalized by $U$ and so $Q(U;D)$ is accurate for $\mathcal{U}$. Let $\tau$ be defined by $\tau(D) = \inf_{U \in \mathcal{U}} Q(U;D)$. Note that $\tau(D)$ just reports the Euclidean distance between $D$ and the closed set $\mathcal{D}_{\mathcal{U}}^T$. It follows that $\tau$ is continuous and further $\tau(D) = || D - \ti{D} ||$ for some $\ti{D} \in \mathcal{D}_{\mathcal{U}}^T$. From the definition of $\mathcal{D}_{\mathcal{U}}^T$ there exists a utility function $\ti{U} \in \mathcal{U}$ so that $\ti{D}$ is rationalized by $U$ and so $\tau(D) = Q( \ti{U}; D )$. Thus, we see that $Q$ is minimizable. Now, as $Q$ is accurate and minimizable Proposition \ref{prop:accurate} guarantees that $\tau$ is accurate.
\end{proof}

\begin{remark}
	The proof that statement 1 implies statement 3 in Proposition \ref{prop:good-news} is constructive. Specifically, let $Q(U;D)$ report the Euclidean distance between $D$ and the collection of $T$ observation datasets which are rationalized by $U$. In the proof of Proposition \ref{prop:good-news} we show that if the collection $\mathcal{D}_{\mathcal{U}}^T$ is closed then the loss function $Q(U;D)$ is accurate and minimizable for $\mathcal{U}$ and generates a goodness-of-fit indices for $\mathcal{U}$ via \eqref{eq:tau-gen} which is continuous and accurate for $\mathcal{U}$. 
\end{remark}

We can use Proposition \ref{prop:good-news} to show that certain classes of utility functions possess loss functions which are minimizable and accurate, and generate goodness-of-fit indices via \eqref{eq:tau-gen} which are continuous and accurate. We identify three such classes below.  

\paragraph{Objective Concave EU.} Let $\pi = (\pi_1, \pi_2, \ldots,\pi_L) \in (0,1]^L$ be a fixed probability vector (i.e.\ $\sum_{\ell=1}^L \pi_{\ell} = 1$). A utility function $U: \rlp \rightarrow \mathbb{R}$ is an objective concave expected utility (OCEU) function if $U(\bs{q}) = \sum_{\ell=1}^L \pi_{\ell} u(\bs{q}_{\ell})$ where $u: \mathbb{R}_+ \rightarrow \mathbb{R}$ is continuous, concave, and increasing.

Given a dataset $D = (\bs{q}^t, \bs{p}^t)_{t \leq T}$ a finite sequence of pairs $( q_{\ell_i}^{t_i}, q_{\ti{\ell}_i}^{\ti{t}_i} )_{i \leq I}$ is called an \emph{objective test sequence} if (i) $q_{\ell_i}^{t_i} > q_{\ti{\ell}_i}^{\ti{t}_i}$ for all $i \leq I$ and (ii) each observation index $t$ appears the same number of times as $t_i$ as it appears as $\ti{t}_i$. It is known that a dataset $D$ can be rationalized by an objective concave expected utility function if and only if, for any objective test sequence $( q_{\ell_i}^{t_i}, q_{\ti{\ell}_i}^{\ti{t}_i} )_{i \leq I}$ 
		\begin{equation} \label{eq:OCEU}
			\prod_{i=1}^I \dfrac{ \pi_{\ti{\ell}_i} p_{\ell_i}^{t_i} }{ \pi_{\ell_i} p_{\ti{\ell}_i}^{\ti{t}_i} } \leq 1
		\end{equation}
(see \citet{kubler-selden-wei14} and \citet{echenique-imai-saito23}).

\paragraph{Subjective Concave EU.} This is similar to the objective expected utility case but the probability vector is not given. So, a utility function $U: \rlp \rightarrow \mathbb{R}$ is a subjective concave expected utility (SCEU) function if $U(\bs{q}) = \sum_{\ell=1}^L \pi_{\ell} u(q_{\ell})$ where $\pi = (\pi_1, \ldots, \pi_L) \in (0,1]^L$ is a probability vector (i.e.\ $\sum_{\ell=1}^L \pi_{\ell} = 1$) and $u: \mathbb{R}_{+} \rightarrow \mathbb{R}$ is continuous, concave, and increasing. 

Given a dataset $D = (\bs{q}^t, \bs{p}^t)_{t \leq T}$, a finite sequence of pairs $( q_{\ell_i}^{t_i}, q_{\ti{\ell}_i}^{\ti{t}_i} )_{i \leq I}$ is called a \emph{subjective test sequence} if (i) $( q_{\ell_i}^{t_i}, q_{\ti{\ell}_i}^{\ti{t}_i} )_{i \leq I}$ is an objective test sequence and (ii) each $\ell$ appears the same number of times as $\ell_i$ as it appears as $\ti{\ell}_i$. \citet{echeniquekota15} show that $D$ can be rationalized by a subjective concave expected utility function if and only if, for any subjective test sequence $( q_{\ell_i}^{t_i}, q_{\ti{\ell}_i}^{\ti{t}_i} )_{i \leq I}$ 
		\begin{equation} \label{eq:SCEU}
			\prod_{i=1}^I \dfrac{ p_{\ell_i}^{t_i} }{ p_{\ti{\ell}_i}^{\ti{t}_i} } \leq 1
		\end{equation}

\paragraph{Homothetic Utility.} A utility function $U: \rlp \rightarrow \mathbb{R}$ is homothetic if it satisfies $U(\bs{q}) \geq U(\bs{q}')$ if and only if $U(k \bs{q}) \geq U(k \bs{q}')$ for all $\bs{q}, \bs{q}' \in \rlp$ and all $k > 0$. It is shown in \citet{varian83} that $D$ can be rationalized by a continuous, increasing, and homothetic utility function if and only if, for all sequences $t_1, t_2, \ldots, t_K$ with $t_1 = t_K$
		\begin{equation} \label{eq:homothetic}
			\prod_{ k =1 }^{K-1} \dfrac{ \bs{p}^{t_k} \cdot \bs{q}^{t_{k+1}} }{ \bs{p}^{t_k} \cdot \bs{q}^{t_k} } \geq 1  
		\end{equation}

\begin{proposition} \label{prop:various}
	For the class of OCEU, SCEU, or homothetic utility functions, there exists a minimizable and accurate loss function which generates an accurate and continuous goodness-of-fit index.
\end{proposition}
\begin{proof}
	Note that the conditions in \eqref{eq:OCEU}, \eqref{eq:SCEU}, and \eqref{eq:homothetic} involve weak inequalities and continuous functions of the purchase data and are thus preserved under limits. Also, note that a sequence of pairs $( q_{\ell_i}^{t_i}, q_{\ti{\ell}_i}^{\ti{t}_i} )_{i \leq I}$ which is not an objective (subjective) test sequence will have a limit which is also not an objective (subjective) test sequence. These observations imply that the set in item 1 of Proposition \ref{prop:good-news} is closed for these families of utility functions and, therefore, Proposition \ref{prop:good-news} establishes the result.
\end{proof}

\section{An Empirical Demonstration} \label{sec:empirical}

\subsection{The data}

In this section we demonstrate how we can calculate robust preferences (in the sense explained in Section \ref{sec:robust-prefs}) over healthy and unhealthy beverages for households in the ``Complete Journey'' dataset published by the market research firm DunnHumby. This dataset consists of scanner data on approximately 2.6 million purchases made by 2,500 households in the U.S. over a two-year period. Our study considers the welfare impact of a reduction in soda consumption. While the original dataset has thousands of products over many categories, we focus our analysis on non-alcoholic beverages which we aggregate into two products: (i) soda and (ii) non-soda beverages (such as fruit juices).  Product and price aggregation was performed using standard methods which we detail in the online appendix. Each household is represented by a dataset of the form $((q^t_s,q^t_o), (p^t_s,p^t_o))_{t\leq T}$, where $q^t_s$ and $p^t_s$ is the quantity and price of soda at observation $t$ and $q^t_o$ and $p^t_o$ is the quantity and price of other-beverages at observation $t$. Each observation $t$ covers consumption over a period of two weeks.
  
Our analysis is performed on a panel of 267 moderately rational households with significant soda consumption.  We arrive at this panel principally by dropping households with only infrequently observed purchases or which consume too little soda; in a small number of cases, households were also removed for scoring very poorly on the Varian index. A complete description of our filtering procedure is in the online appendix.

The number of observations for each of these 267 households vary between 39 and 52, with a median of 44. On average, these households spend \$8.03 on soda every two weeks (with a standard deviation across households of \$6.63) and spend \$14.35 on non-soda drinks every two weeks (with a standard deviation across households of \$8.80). Only 16 of these households passed GARP (a pass rate of 6\%) and so Proposition \ref{prop:popular-loss-functions} tells us that 94\% of households in our sample do not have best-fitting utility functions when using many of the standard loss functions.

We shall compute some robust preferences for these households using three indices: the Afriat index, the Varian index (with additive aggregator), and the Houtman-Maks index. For a first exercise we compute the rank correlations between these indices which we display in Table \ref{tab:correlation}.\footnote{To compute the Varian and Houtman-Maks indices we use the mixed integer programming approach from \citet{demuynck-rehbeck23}.}
\begin{table}[ht]
	\centering
	\caption{Spearman’s Rank Correlation across indices}
	\label{tab:correlation}
	\begin{tabular}{l|cccc}
		
		& $ A $ & $ V $ & $ H $ &  \\
		\hline
		$ A $ & 1.000  &       &       &       \\
		$ V $ & 0.922 & 1.000  &       &       \\
		$ H $ & 0.436  & 0.640  & 1.000  &       \\
		
	\end{tabular}
\end{table}
Each entry in Table \ref{tab:correlation} presents the rank correlation between two indices. It is interesting to note that, while the Afriat and Varian indices appear highly correlated (rank correlation 0.922) the Houtman-Maks index is considerably less correlated with the other two indices (rank correlation 0.4 and 0.6 for the Afriat and Varian indices, respectively). This finding seems sensible. Both the Afriat and Varian indices are based on the shared concept of wasted expenditure while Houtman-Maks is based on the idea that some choices are uninformative mistakes which should be removed. Additional, as we have seen in Proposition \ref{prop:popular-indices}, the Afriat and Varian indices share the theoretical properties of being continuous and inaccurate whereas Houtman-Maks is accurate but not continuous. 

\subsection{Robust Preferences} 

Our goal is to estimate the extent to which healthy beverage consumption must be increased to compensate consumers for a decrease in soda consumption of 25\%.  Suppose we order the bundles purchased by a single houeshold  $\bs{q}^1, \bs{q}^2, \ldots, \bs{q}^T$ by increasing soda consumption, i.e., $q_s^1 \leq q_s^2 \leq \ldots \leq q_s^T$. We denote the bundle with the median soda consumption by $\bs{q}^M=(q_s^M,q_o^M)$. 

We estimate the increase in non-soda drink consumption required to compensate households for a decrease in soda consumption of 25\%. For each $k \in [0,1]$ we compare the median bundle $\bs{q}^M$ to the healthier bundle 
$$\bs{q}^{M,k} = \left( \tfrac{3}{4} q_s^M, (1+k) q_o^M \right)$$ 
which consists of a reduction in soda consumption of 25\% and an increase in non-soda consumption of $k \%$ (compared to the median bundle). Note that when $k = 0$ then $\bs{q}^M > \bs{q}^{M,k}$ and so robust preferences require $\bs{q}^M \ \mathcal{R}_Q \ \bs{q}^{M,0}$ for any loss function $Q$.

\begin{figure}[h]
	\centering
	\includegraphics[width=.7\textwidth]{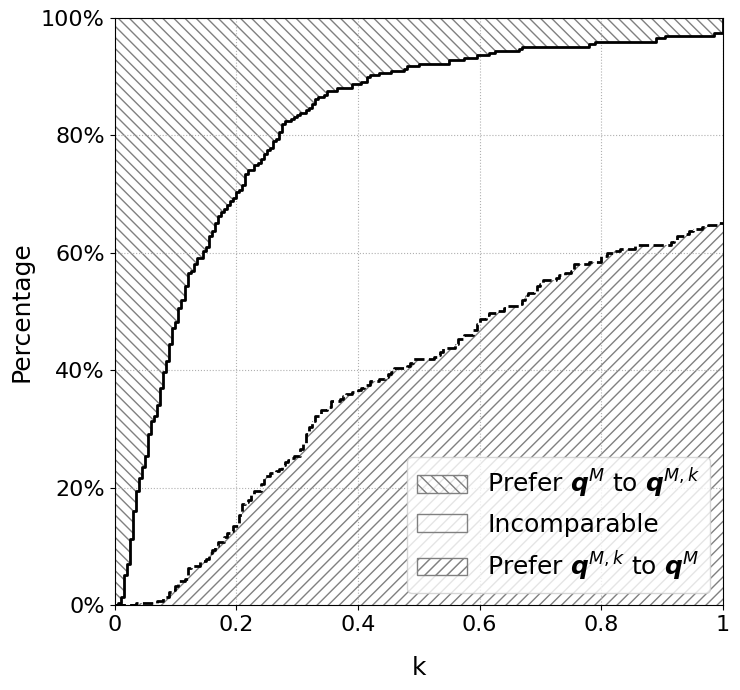}
	\caption{Weak and strong compensation levels for a 25\% reduction in soda consumption using the Varian index} \label{fig:switch-V}
\end{figure}
\vspace{0.15in}

Each household has two values of $k$ which we shall be interested in estimating. 
\begin{itemize}
\item [(i)] the {\bf weak compensation level} $k_w \in [0,1]$ is the smallest (infimum) number $k$ so that $\bs{q}^M \ \cancel{ \mathcal{R}_Q } \ \bs{q}^{M,k}$, i.e., $k_w$ is the smallest number so that we cannot conclude that the healthy bundle $\bs{q}^{M,k}$ is dominated by the median bundle $\bs{q}^M$. This number $k_w$ represents a lower bound on the amount of non-soda consumption which must be given to the household to compensate for the reduction in soda consumption; 
\item[(ii)] the {\bf strong compensation level} $k_s \in [0,1]$ is the smallest (infimum) number $k$ so that $\bs{q}^{M,k} \ \mathcal{R}_Q \ \bs{q}^M$. This number is an upper bound on the amount of non-soda consumption which must be given to the household to compensate for the reduction in soda consumption. Obviously $k_s \geq k_w$.
\end{itemize}
\vspace{10pt}

Figure \ref{fig:switch-V} displays the weak and strong compensation levels for the Varian index.\footnote{Details for computing the robust preference relation using the Varian loss function are in the online appendix.} The horizontal axis of this figure displays different compensation levels $k$ (i.e.\ different percentage increases in non-soda beverage consumption) while the vertical axis reports the percentage of households belonging to various categories. The figure is divided into three regions. The top region (shaded with diagonal falling lines) depicts households who prefer the baseline bundle $\bs{q}^{M}$ to the healthier bundle $\bs{q}^{M,k}$ (i.e.\ $\bs{q}^M \ \mathcal{R}_Q \ \bs{q}^{M,k}$). So for instance, at $k = 0.4$, we observe that roughly 10\% of households (the length of the upper region at $k= 0.4$) are found to prefer the unhealthy bundle $\bs{q}^M$ to the healthier $\bs{q}^{M,k}$. The bottom region depicts households who prefer the healthier bundle to the baseline bundle (i.e.\ $\bs{q}^{M,k} \ \mathcal{R}_Q \ \bs{q}^{M}$).  Taking $k=0.4$ we see that roughly 40\% of households prefer the healthier bundle $\bs{q}^{M,k}$ to the unhealthy bundle. The white region in the middle depicts households for whom comparison is not possible (recall that $\mathcal{R}_Q$ is not a complete relation). Again, considering $k=0.4$ we see that for roughly 50\% of households (the length of the white region at $k=0.4$), there is neither a robust preference for $\bs{q}^M$ nor a robust preference for $\bs{q}^{M,k}$.
 
The solid line in the figure (i.e.\ the boundary between the top and middle sections) represents the cumulative distribution function (CDF) of $k_w$. Note that this line does not end at 1.0 as there is no reason why a 100\% increase in non-soda consumption (the compensation level corresponding to the rightmost points in the figure) should be enough to compensate for a 25\% reduction in soda consumption. As we noted, $\bs{q}^M > \bs{q}^{M,0}$ and thus no household can have weak compensation level 0. That said, the figure depicts a sharp rise in the CDF of $k_w$ immediately to the right of 0 suggesting that a small amount of compensation is enough to move many households out of the region where the baseline bundle is robustly preferred to the healthier bundle. Indeed, an increase in non-soda beverage consumption by 10\% suffices to weakly compensate over 50\% of households for the considered 25\% reduction in soda consumption and a 25\% increase in non-soda beverage consumption is sufficient to weakly compensates about 80\% of households.

The dashed line in the figure (i.e.\ the boundary between the middle and bottom section) depicts the CDF for strong compensation level $k_s$. Of course, this CDF must lie below the CDF of the weak compensation level $k_w$. In contrast to the CDF for weak compensation we note that the CDF of $k_s$ rises at a fairly linear rate with a slope of about 0.6. In other words, for every 1\% increase in the compensation level in non-soda beverage consumption we win over an additional 0.6\% of households who now robustly prefer the healthier bundle to the baseline consumption level. 

For each household, there is typically a significant difference between the weak and strong compensation levels.  This is unsurprising, since (loosely speaking), the former represents that level of $k$ at which {\em there exists} a well-behaved utility function consistent with the data with a preference for the healthy bundle whereas the latter guarantees that {\em all} well-behaved utility functions consistent with the data have a preference for the healthy bundle.  In the population of households, the 25th, 50th, and 75th percentile values of $k$ for weak compensation are 0.05, 0.1, and 0.23 respectively; the corresponding values for strong compensation are 0.29, 0.63, and 1.44.


Figure \ref{fig:switching} depicts the strong and weak compensation levels for the Afriat and the Houtman-Maks indices.\footnote{We discussed how to compute robust preferences using the Afriat loss function in Section \ref{sec:robust-prefs} and further details are contained in the appendix. Computing robust preferences using the Houtman-Maks loss function is discussed in the online appendix.} From inspection, we see that the figures for the Varian, Afriat, and Houtman-Maks indices are very similar. To investigate more deeply, we calculate the correlation between these compensation levels in Table \ref{tab:cutoffs}. This table shows that the strong compensation levels using different loss functions are all highly correlated, though somewhat lower between Houtman-Maks and the other two loss functions.  The weak compensation levels using the Afriat and Varian loss functions are highly correlated but the correlation between the weak compensation level using Houtman-Maks and either the Afriat or Varian loss functions is somewhat lower.  This roughly aligns with the message in Table \ref{tab:correlation} which suggests that the Houtman-Maks index can respond quite differently to GARP violations compared to the Afriat and Varian indices.

\begin{figure}[h]
	\centering
	\begin{subfigure}[b]{0.49\textwidth}
		\centering
		\includegraphics[width=\textwidth]{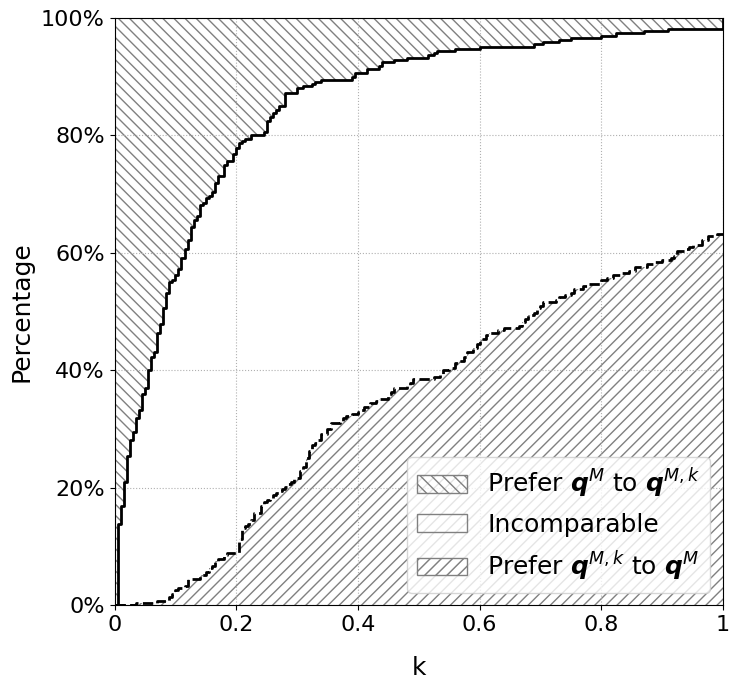}
		\caption{Afriat}
		\label{fig:switch-A}
	\end{subfigure}
	\hfill
	\begin{subfigure}[b]{0.49\textwidth}
		\centering
		\includegraphics[width=\textwidth]{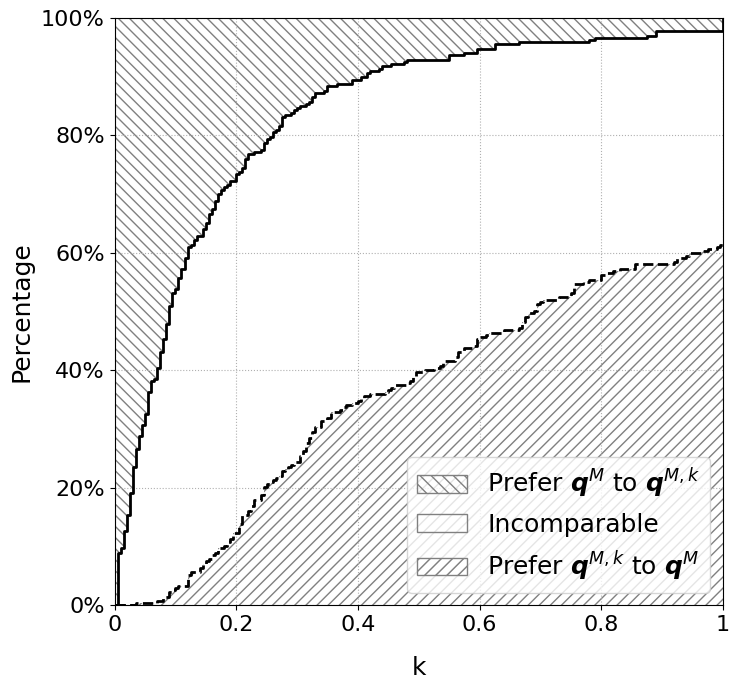}
		\caption{Houtman-Maks}
		\label{fig:switch-H}
	\end{subfigure}
	\caption{Weak and strong compensation levels for a 25\% reduction in soda consumption using the Varian Index}
	\label{fig:switching}
\end{figure}


\begin{table}[ht]
	\centering
	\caption{Correlations across indices weak and strong compensation levels}
	\label{tab:cutoffs}
	\begin{tabular}{l|ccc | ccc}
		& \multicolumn{3}{c|}{weak comp.} & \multicolumn{3}{c}{strong comp.} \\
		& $ A $ & $ V $ & $ H $ & $A$ & $V$ & $H$ \\
		\hline
		$ A $ & 1.000  &        &        & 1.000 &       & \\
		$ V $ & 0.996  & 1.000  &        & 0.975 & 1.000 &     \\
		$ H $ & 0.791  & 0.776  & 1.000  & 0.907 & 0.908 & 1.000 \\
		
	\end{tabular}
\end{table}

While the compensation levels are highly correlated across different loss functions, it may be that one loss function generally provides sharper welfare conclusions than the others. For a loss function $Q$ let $I_Q = [ k_w, k_s ]$ where $k_w$ and $k_s$ are the weak and strong compensation levels for some fixed household. Each $k \in ( k_w, k_s )$ represents a compensation level whose welfare content is unclear in the sense that $\bs{q}^{M} \ \cancel{ \mathcal{R}_{Q}} \ \bs{q}^{M,k}$ and $\bs{q}^{M,k} \ \cancel{\mathcal{R}_{Q}} \ \bs{q}^{M}$. For this reason, we say that loss function $Q$ is \emph{sharper} than $\ti{Q}$ if $I_Q \subseteq I_{\ti{Q}}$ and $I_{\ti{Q}} \not \subseteq I_Q$.

We find that the Varian and Houtman-Maks loss functions consistently provide sharper welfare bounds than the Afriat loss function. Specifically, the Varian loss function is sharper than the Afriat loss function for 91.76\% of households and the Houtman-Maks loss function is sharper than the Afriat loss function for 74.16\% of households. The Varian loss function appears to have a notable edge over Houtman-Maks as Varian is sharper for 19.85\% of households while Houtman-Maks is sharper for only 1.87\% of households (Varian and Houtman-Maks provide the exact same welfare bounds for 76.78\% of households and the bounds are incomparable for 1.50\% of households.)

\subsection{Takeaways}

We find strong correlations between the welfare predictions made by the Afriat, Varian, and Houtman-Maks loss functions.  For the average household, the weak compensation level is modest but the strong compensation level is considerably higher, i.e., only a moderate amount of compensation (in the form of `other beverages')  is required to give the household a healthy bundle (with a 25\% reduction in soda consumption) that is not robustly dominated by the original bundle, but considerably more compensation is required to ensure that the healthy bundle is robustly preferred to the original bundle.  The Varian loss function appears to be best performing in the sense that, more often than not, it gives the narrowest gaps between the weak and strong compensation levels.


\section{Conclusion} \label{sec:conclusion}

We have shown in this paper that, for the class of well-behaved utility functions, there are no continuous and accurate goodness-of-fit indices, nor are there any minimizable and accurate loss functions which generate continuous goodness-of-fit indices. The Afriat, Varian, Swaps, and LS loss functions are accurate but not minimizable; their corresponding indices are continuous and essentially accurate but not accurate.  The Houtman-Maks index loss function is accurate and minimizable but leads to a discontinuous goodness-of-fit index.  We introduced the robust preference relation $\mathcal{R}_Q$ and argued that it is suitable for making welfare statements even when the loss function is not minimizable (in other words, even when there is no best-fitting utility function according to that loss function). In an empirical demonstration, we calculated robust preference relations based on the Afriat, Varian, and Houtman-Maks loss functions. Lastly, we characterized those classes of utility functions which (unlike the class of well-behaved utility functions) {\em do} admit accurate and minimizable loss functions generating accurate and continuous goodness-of-fit indices.


\appendix

\section{Appendix} \label{sec:appendix}

The appendix is organized as follows. In Section \ref{sec:props} we prove most of the unproved propositions. Section \ref{sec:afriat-appendix} proves Proposition \ref{prop:afriat-welfare}. Section \ref{sec:e^*-calc} shows how to calculate the Afriat index exactly.

\subsection{Proof of Propositions \ref{prop:badnews2}, \ref{prop:popular-indices}, \ref{prop:essacc-indices}, \ref{prop:popular-loss-functions}, and \ref{prop:uniform-pref}} \label{sec:props}

First we present some prerequisite results.

\subsubsection{Lemmas} \label{sec:NOQ}

We will need a version of Theorem 2 in \citet{nishimura-oK-quah17}. A \emph{choice dataset} is a a finite collection of pairs $D = (\bs{q}^t, B^t)_{t \leq T}$ where $\bs{q}^t \in B^t$ and $B^t \subseteq \mathbb{R}_{+}^L$ is a compact set. We say that $D$ is \emph{rationalized} by $U$ if, for all $t$, we have $U(\bs{q}^t) \geq U(\bs{q})$ for all $\bs{q} \in B^t$. We write $\bs{q}^t \ R \ \bs{q}^s$ if there exist some $\tbs{q} \in B^t$ so that $\tbs{q} \geq \bs{q}^s$ and $\bs{q}^t \ P \ \bs{q}^s$ if there exists $\tbs{q} \in B^t$ so that $\tbs{q} > \bs{q}^s$. We refer to $R$ and $P$ as the direct and strict direct revealed preference relations for $D$. We say that $D$ satisfies \emph{cyclical consistency} if, for all $t_1, t_2, \ldots, t_K$ we have $\bs{q}^{t_1} \ R \ \bs{q}^{t_2} \ R \ \ldots \ R \ \bs{q}^{t_K}$ implies not $\bs{q}^{t_K} \ P \ \bs{q}^{t_1}$. The following is a version of Theorem 2 in \citet{nishimura-oK-quah17}.
\begin{lemma}[NOQ, 2017] \label{lemma:NOQ}
	The choice dataset $D = (\bs{q}^t, B^t)_{t \leq T}$ is rationalized by a well-behaved utility function if and only if it satisfies cyclical consistency.
\end{lemma}

We move on to proving an additional helpful lemma. Recall that $A$, $V$, and $S$ are the Afriat, Varian, and Swaps loss functions, respectively. Also, recall that $f$ is the function in the definition of the Varian loss function and $\mu$ is the Lebesgue measure.
\begin{lemma} \label{lemma:gof-rep}
	Let $D = (\bs{q}^t, \bs{p}^t)_{t \leq T}$ be a purchase dataset and let $\mathcal{Q} = \{ \bs{q}^t \}_{t \leq T}$. Let $\mathcal{P}_D$ be the collection of all transitive and complete binary relations $\succeq$ on $\mathcal{Q}$ which satisfy $\bs{q}^t \geq (>) \ \bs{q}^s$ implies $\bs{q}^t \succeq (\succ) \ \bs{q}^s$.\footnote{As usual $\bs{q}^t \succ \bs{q}^s$ means $\bs{q}^t \succeq \bs{q}^s$ and $\bs{q}^s \not\succeq \bs{q}^t$.} For $\succeq \ \in \mathcal{P}_D$ let $\bs{e}^{\succeq} = ( e_1^{\succeq}, e_2^{\succeq}, \ldots, e_T^{\succeq} )$ be defined by 
	\begin{equation*}
		e_{t}^{\succeq} = \min_{s \in \{s': \bs{q}^{s'} \succeq \bs{q}^t \} } \dfrac{ \bs{p}^t \cdot \bs{q}^s }{ \bs{p}^t \cdot \bs{q}^t }
	\end{equation*}
	Let $A(\succeq;D) = \min_t e_t^{\succeq}$, $V(\succeq;D) = f(\bs{e}^{\succeq})$, and $S(\succeq; D) = \sum_{t=1}^T \mu( \{ \bs{q} \in B(\bs{p}^t, \bs{p}^t \cdot \bs{q}^t): \exists s \text{ s.t.\ } \bs{q} \geq \bs{q}^s \succeq \bs{q}^t \} )$. Then, for any $Q \in \{ A, V, S \}$,
	\begin{IEEEeqnarray}{rcCcl}
		\inf_{U \in \mathcal{U}_{WB}} Q(U;D) & = & \inf_{\succeq \in \mathcal{P}_D} Q(\succeq; D) \label{eq:gof-rep}
	\end{IEEEeqnarray} 
\end{lemma}
\begin{proof}
	Let $Q \in \{A, V, S\}$ and let $\tau$ be defined by \eqref{eq:tau-gen}. It is obvious that $\inf_{U \in \mathcal{U}_{WB}} Q( U; D ) \geq \inf_{\succeq \in \mathcal{P}_{D}} Q( \succeq; D )$ and so, to complete the proof, we must show that $\inf_{\succeq \in \mathcal{P}_D} Q( \succeq; D ) \geq \inf_{U \in \mathcal{U}_{WB}} Q( U; D )$.
	
	As $\mathcal{P}_D$ contains only finitely many elements the infimum in \eqref{eq:gof-rep} is attained by some $\succeq^* \ \in \mathcal{P}_D$. For each $\varepsilon > 0$ and each $t$ let $B_{\varepsilon}^t = \{ \bs{q} \in B(\bs{p}^t, \bs{p}^t \cdot \bs{q}^t): \forall s \in \{1,2\ldots, T\} \text{ s.t.\ } \bs{q}^t \succeq^* \bs{q}^s \text{ we have } \bs{q} \ \cancel{\gg} \ (1-\varepsilon) \bs{q}^s \} \cup \{ \bs{q}^t\}$. Let $D_{\varepsilon} = ( \bs{q}^t, B_{\varepsilon}^t )_{t \leq T}$. It is easy to see that $D_{\varepsilon}$ satisfies cyclical consistency and so, by Theorem 2 in \citet{nishimura-oK-quah17} (Lemma \ref{lemma:NOQ} above), there exists a well-behaved $U_{\varepsilon}$ which rationalizes $D_{\varepsilon}$. It is easy to see (if $Q$ is $\swaps$ then use the continuity-from-below property of measures) that $Q(U_{\varepsilon};D) \rightarrow Q(\succeq^*; D)$ as $\varepsilon \rightarrow 0$ and so $\inf_{U \in \mathcal{U}_{WB}} Q(U;D) \leq \inf_{\succeq \in \mathcal{P}_D} Q(\succeq; D)$ and the proof is complete.
\end{proof}

\subsubsection{Proofs of the propositions}

\begin{proof}[Proof of Proposition \ref{prop:badnews2}.]
	We must show that for every dataset $D$ which contains no strong RP-cycles we can perturb the bundles so that the perturbed dataset is rationalizable. Clearly this perturbation is possible for datasets with one observation (indeed one observation datasets are always rationalizable and so the dataset need not be perturbed at all). So we assume that such perturbations are possible for all datasets with $T$ observations and will show that it is also possible for datasets with $T+1$ observations.
	
	Let $D$ be a $T+1$ observation dataset without strong RP-cycles. If $D$ contains no weak RP-cycles then the result is trivial so suppose $D$ contains some weak RP-cycles. So, for any sequence $t_1,t_2,\ldots, t_K$ with $t_1 = t_K$ there exists $k < K$ so that $\bs{q}^{t_k} \ \cancel{P} \ \bs{q}^{t_{k+1}}$. Thus, there must exist an observation $s$ such that $\bs{q}^s \ \cancel{ P } \ \bs{q}^t$ for all $t$ because if no such observation existed then we could find a sequence $t_1, t_2, \ldots , t_K$ with $t_1 = t_K$ which satisfied $\bs{q}^{t_k} \ P \ \bs{q}^{t_{k+1}}$ for all $k < K$.
	
	So, remove the observation $(\bs{q}^s, \bs{p}^s)$ from the set of $T+1$ observations. The remaining collection $D'$ has $T$ observations and, by our induction assumption, there is a perturbed dataset $\ti{D}'=\{( \tbs{q}^t,  \bs{p}^t)\}_{t\leq T}$ that is rationalizable and so there is a well-behaved utility function $U: \rlp \rightarrow \mathbb{R}$ which rationalizes $\ti{D}'$.
	
	Let $\hat{\bs{q}}^s$ be a bundle on $B(\bs{p}^s, \bs{p}^s \cdot \bs{q}^s)$ that maximizes utility function $U$. We claim that the dataset formed by appending $(\tbs{q}^s, \bs{p}^s )$ to $\ti{D}'$ is still rationalizable, where $\tbs{q}^s = \alpha \bs{q}^s + (1-\alpha) \hat{\bs{q}}^s$ for some $\alpha\in (0,1)$.
	
	Suppose not. Then $(\tbs{q}^s, \bs{p}^s)$ must be part of a GARP-violating cycle. Then there is another observation $r$ such that $\bs{p}^r \cdot \tbs{q}^r \geq \bs{p}^r \cdot \tbs{q}^s$.  Since $\bs{p}^r \cdot \tbs{q}^r \leq \bs{p}^r \cdot \bs{q}^s$ by our choice of observation $s$, we have $\bs{p}^r \cdot \tbs{q}^r \geq \bs{p}^r \cdot \hat{\bs{q}}^s$. But this means that if we replace $(\bs{p}^s, \tbs{q}^s)$ with $(\bs{p}^s, \hat{\bs{q}}^s)$, there will still be a GARP-violating cycle. But that's impossible since $\ti{D}'$ with observation $(\bs{p}^s, \hat{\bs{q}}^s)$ appended can be rationalized by $U$. By choosing $\alpha$ arbitrarily close to $1$ we find the desired perturbed dataset. Thus, any dataset $D$ without strong RP-cycles can be perturbed by an arbitrarily small amount in order to restore GARP. 
\end{proof}

\begin{proof}[Proof of Proposition \ref{prop:popular-indices}.]
	Houtman Maks is obviously accurate from its definition. The continuity of the Afriat, Varian, and Swaps indices follows from Lemma \ref{lemma:gof-rep}. Let $\mathcal{D} \subseteq ( (\rlp \backslash \{\bs{0}\}) \times \mathbb{R}_{++})^T$ be the collection of all $T$ observation datasets which satisfy GARP and let $\bar{\mathcal{D}}$ denote the topological closure of $\mathcal{D}$. For two $T$ observation datasets $D = (\bs{q}^t,\bs{p}^t)_{t \leq T}$ and $\ti{D} = (\tbs{q}^t, \tbs{p}^t)_{t \leq T}$ let $d( D, \ti{D} ) = \sum_{t=1}^T (\bs{q}^t - \tbs{q}^t)' \bs{W}_D^t (\bs{q}^t - \tbs{q}^t)$ where $\bs{W}_D^t$ is the matrix which appears in the definition of the LS loss function. Let $\Gamma( (\bs{q}^t, \bs{p}^t)_{t \leq T} ) = \{ (\tbs{q}^t, \bs{p}^t)_{ t \leq T} \in \bar{\mathcal{D}}: \tbs{q}^t \in B(\bs{p}^t, \bs{p}^t \cdot \bs{q}^t) \text{ for all } t \}$. Note that $\ls(D) = \inf_{ \ti{D} \in \Gamma(D) } d( D, \ti{D} )$. Now, the continuity of $\ls$ follows from the Berge Maximum Theorem.

	The negative results for the indices now follow from Proposition \ref{prop:badnews2}.
\end{proof}

\begin{proof}[Proof of Proposition \ref{prop:essacc-indices}.]
	Given Proposition \ref{prop:badnews2} we see that we need only show that the indices take on positive values when the data contains a strong GARP violation. This is essentially true by definition for the $\ls$ index and so we focus on the remaining indices.
	
	So, suppose $D$ contains a strong cycle. That is, there exists $t_1,t_2,\ldots, t_K$ so that
	\begin{equation} \label{eq:strict-cycle}
		\bs{q}^{t_1} \ P \ \bs{q}^{t_2} \ P \ \bs{q}^{t_3} \ P \ \ldots \ P \ \bs{q}^{t_K} \ P \ \bs{q}^{t_1}
	\end{equation}
	Also, let $t_{K+1} = t_1$. Of course, it is impossible to have a function $U$ satisfy $U(\bs{q}^{t_1}) > U(\bs{q}^{t_2}) > \ldots > U( \bs{q}^{t_K} ) > U( \bs{q}^{t_1} )$ and so, 
	\begin{equation} \label{eq:bad-k-exists}
		\forall U \in \mathcal{U}_{WB} \text{ there exists } k \text{ s.t. } U(\bs{q}^{t_{k+1}}) \geq U(\bs{q}^{t_k})
	\end{equation}
	Take 
	\begin{equation*}
		e^* = \max_{k} \ \dfrac{ \bs{p}^{t_k} \cdot \bs{x}^{t_{k+1}} }{ \bs{p}^{t_k} \cdot \bs{x}^{t_{k}} }
	\end{equation*}
	Because \eqref{eq:strict-cycle} holds we have $e^* < 1$ and from \eqref{eq:bad-k-exists} we have $0 < 1-e^* \leq A(U;D)$ for all $U$. So, $A$ is essentially accurate.  A similar argument demonstrates that $0 < V(D)$ and so $V$ is essentially accurate. Finally, from \eqref{eq:strict-cycle} and \eqref{eq:bad-k-exists}
	\begin{equation*}
		S(D) \geq \min_k \ \mu \Big\{ \bs{q} \in B(\bs{p}^{t_k}, \bs{p}^{t_k} \cdot \bs{q}^{t_k}): U(\bs{q}) \geq U(\bs{q}^{t_{k+1}}) \Big\} > 0
	\end{equation*}
	and so $S$ is also essentially accurate.
\end{proof}

\begin{proof}[Proof of Proposition \ref{prop:popular-loss-functions}.]
	In what follows, recall that $A$, $V$, $S$, and $\ls$ refers to the Afriat, Varian, Swaps, and LS loss functions and indices, respectively. We first prove that the loss functions are accurate and then we prove the empty argmin results. \vspace{7pt}
	
	\noindent \emph{Accuracy of loss functions.} Let $Q \in \{A, V, S, \ls\}$. If $D$ is rationalized by $U \in \mathcal{U}_{WB}$ then it is easy to see that $Q(U;D) = 0$. On the other hand, suppose $D$ is not rationalized by $U \in \mathcal{U}_{WB}$. As $U$ does not rationalize $D$ there exist some observation $t$ and bundle $\tbs{q} \in \mathbb{R}_{+}^L$ so that $U( \tbs{q} ) > U(\bs{q}^t)$ and $\bs{p}^t \cdot \bs{q}^t \geq \bs{p}^t \cdot \tbs{q}$. As $U$ is continuous there exists $\tbs{q}'$ so that $U( \tbs{q}' ) > U(\bs{q}^t)$ and $\bs{p}^t \cdot \bs{q}^t > \bs{p}^t \cdot \tbs{q}'$. It is now easy to show that $Q(U;D) > 0$ using the existence of $\tbs{q}'$.  \vspace{7pt}
	
	\noindent \emph{Varian and Swaps empty argmin.} Let $D = (\bs{q}^t, \bs{p}^t)_{t \leq T}$ be a purchase dataset which does not satisfy GARP and let $\mathcal{Q} = \{\bs{q}^t\}_{t \leq T}$. Let $Q \in \{ V,S \}$. Define $Q(\succeq; D)$ as in Lemma \ref{lemma:gof-rep}. For a contradiction suppose there exists $U^* \in \argmin_{U \in \mathcal{U}_{WB}} Q(U;D)$. Let $\succeq^*$ be the binary relation on $\mathcal{Q}$ generated by $U^*$ in the sense that $\bs{q}^t \succeq^* \bs{q}^s$ iff $U^*( \bs{q}^t ) \geq U^*(\bs{q}^s)$. From Lemma \ref{lemma:gof-rep} it is clear that $Q(U^*; D) = Q(\succeq^*; D)$. We write $\bs{q}^t \ P^* \ \bs{q}^s$ if either (i) $\bs{q}^t \ R \ \bs{q}^s$ and $U^*(\bs{q}^s) > U^*(\bs{q}^t)$ or (ii) $\bs{q}^t \ P \ \bs{q}^s$ and $U^*( \bs{q}^s ) = U^*( \bs{q}^t )$. As $D$ does not satisfy GARP there must be some $\ti{t}$ and $\ti{s}$ so that $\bs{q}^{\ti{t}} \ P^* \ \bs{q}^{\ti{s}}$. Without loss of generality we assume that $\ti{s}$ satisfies $\bs{q}^{\ti{t}} \ \cancel{P^*} \ \bs{q}^{r}$ for all $\bs{q}^r$ satisfying $\bs{q}^{\ti{s}} > \bs{q}^r$ (if $\ti{s}$ does not satisfy this property for some $r$ then just redefine $\ti{s}$ to be $r$ and repeat this process until $\ti{s}$ satisfies the desired property). Now, if case (i) in the definition of $P^*$ holds (i.e.\ $\bs{q}^{\ti{t}} \ R \ \bs{q}^{\ti{s}}$ and $U^*(\bs{q}^{\ti{s}}) > U^*(\bs{q}^{\ti{t}})$) then it is easy to see (as $U^*$ is continuous) that $Q(U^*; D) > Q(\succeq^*; D)$ which is a contradiction. We infer that case (ii) must hold and so $\bs{q}^{\ti{t}} \ P \ \bs{q}^{\ti{s}}$ and $U^*(\bs{q}^{\ti{s}}) = U^*( \bs{q}^{\ti{t}} )$. Now, let $\succeq^{**}$ be any linear order on $\mathcal{Q}$ which satisfies (i) $\bs{q}^t \succeq^{**} \bs{q}^s$ implies $\bs{q}^t \succeq^* \bs{q}^s$ and (ii) $\bs{q}^{\ti{t}} \succeq^{**} \bs{q}^{\ti{s}}$. Consider two cases; (a) $Q = V$ and (b) $Q = S$. Under case (a) let $\bs{e}^{\succeq}$ be defined as in Lemma \ref{lemma:gof-rep} and note that $e_{t}^{\succeq^{**}} \geq e_{t}^{\succeq^{*}}$ for all $t$ and $e_{\ti{t}}^{\succeq^{**}} > e_{ \ti{t} }^{\succeq^{*}}$. Therefore, as $f$ is decreasing, we have $V(\succeq^*;D) > V(\succeq^{**}; D)$ which contradicts equation \eqref{eq:gof-rep} in Lemma \ref{lemma:gof-rep}. We conclude that $\argmin_{U \in \mathcal{U}_{WB}} V(U;D)$ is empty. Under case (b) let $h_t(\succeq) = \mu\{ \bs{q} \in B(\bs{p}^t, \bs{p}^t \cdot \bs{q}^t): \bs{q} \geq \bs{q}^s \succeq \bs{q}^t \}$. Clearly, $h_t(\succeq^{*}) \geq h_t(\succeq^{**})$ for all $t$ and $h_{\ti{t}}(\succeq^*) > h_{ \ti{t} }( \succeq^{**} )$ and thus, $S(\succeq^{*};D) > S(\succeq^{**};D)$ which contradicts equation \eqref{eq:gof-rep} in Lemma \ref{lemma:gof-rep}. We conclude that $\argmin_{U \in \mathcal{U}_{WB}} S(U;D)$ is empty. \\ 
	
	\noindent \emph{LS empty argmin.} Let $D = (\bs{q}^t, \bs{p}^t)_{t \leq T}$ be a purchase dataset which does not satisfy GARP and let $\mathcal{Q} = \{\bs{q}^t\}_{t \leq T}$. Let $\mathcal{D} = \{ ( \tbs{q}^t, \bs{p}^t )_{t \leq T}: \bs{p}^t \cdot \tbs{q}^t = \bs{p}^t \cdot \bs{q}^t, \forall t \}$ (i.e.\ $\mathcal{D}$ is the collection of all $T$ observation purchase datasets with the same budget sets as in $D$). For any two datasets $\ti{D} = (\tbs{q}^t, \bs{p}^t)_{t \leq T}$ and $\hat{D} = (\hbs{q}^t, \bs{p}^t)_{t \leq T}$ in $\mathcal{D}$ we let $\alpha \ti{D} + (1-\alpha) \hat{D} = ( \alpha \tbs{q}^t + (1-\alpha) \hbs{q}^t, \bs{p}^t )_{t \leq T}$ for all $\alpha \in [0,1]$ and we let $d( \ti{D}, \hat{D} ) = \sum_{t=1}^T ( \tbs{q}^t - \hbs{q}^t )' \bs{W}_D^t ( \tbs{q}^t - \hbs{q}^t )$. Now, for a contradiction suppose that $U^* \in \argmin_{U \in \mathcal{U}_{WB}} \ls(U;D)$. As $U^*$ is continuous (and thus, the subset of $\mathcal{D}$ which is rationalized by $U^*$ is closed) there exists some purchase dataset $\ti{D} = (\tbs{q}^t, \bs{p}^t)_{t \leq T}$ which is rationalized by $U^*$ which satisfies $d(D, \ti{D}) = \ls( U^*; D )$. Let $R$ and $P$ be the revealed preference relations for $D$ and let $\ti{R}$ and $\ti{P}$ be the revealed preference relations for $\ti{D}$. Note that $d(D, \ti{D}) > 0$ as $D$ does not satisfy GARP and that $\alpha D + (1-\alpha) \ti{D}$ does not satisfy GARP for every $\alpha \in (0,1]$ and therefore there must exist some $\ti{t}$ and $\ti{s}$ so that $\tbs{q}^{\ti{t}} \ \ti{R} \ \tbs{q}^{\ti{s}}$ and $\bs{q}^{\ti{t}} \ P \ \bs{q}^{\ti{s}}$ (that is, a revealed preference relation which is strict in $D$ is weak in $\ti{D}$). For all $\alpha \in [0,1]$ define $\bs{q}_{\alpha}^t$ so that (i) $\bs{q}_{\alpha}^t = \alpha \bs{q}^{t} + (1-\alpha) \tbs{q}^{t}$ if $t = \ti{s}$ and (ii) $\bs{q}_{\alpha}^t = \tbs{q}^t$ if $t \neq \ti{s}$. Let $D_{\alpha} = (\bs{q}_{\alpha}^t, \bs{p}^t)_{t \leq T}$. Let $R_{\alpha}$ and $P_{\alpha}$ be the revealed preference relations for $D_{\alpha}$. Let $\alpha > 0$ be small enough so that $\bs{q}_{\alpha}^t \ R_{\alpha} \ \bs{q}_{\alpha}^{\ti{s}}$ implies $\tbs{q}^t \ \ti{R} \ \tbs{q}^{\ti{s}}$. It is easy to see that $D_{\alpha}$ satisfies $e$-GARP for all $e < 1$ and thus, by Proposition \ref{prop:AEI}, we have $A(D_{\alpha}) = 0$. As the Afriat index is essentially accurate (Proposition \ref{prop:essacc-indices}) this means $\inf_{U \in \mathcal{U}_{WB}} \ls( U; D_{\alpha} ) = 0$ and, as $d( D, D_{\alpha} ) < d( D, \ti{D} )$ we attain a contradiction and so $\argmin_{U \in \mathcal{U}_{WB}} \ls(U;D)$ is empty. \vspace{7pt}
	
	\noindent \emph{Afriat generic empty argmin.} Let $\mathcal{D}$ be the collection of all $T$ observation purchase datasets which do not satisfy GARP. Let $\mathcal{D}' \subseteq \mathcal{D}$ be the collection of all $T$ observation purchase dataset $D = (\bs{q}^t, \bs{p}^t)_{t \leq T}$ which (i) do not satisfy GARP and (ii) satisfy $\bs{p}^t \cdot \bs{q}^s / (\bs{p}^t\cdot \bs{q}^t) \neq \bs{p}^s \cdot \bs{q}^r / ( \bs{p}^s \cdot \bs{q}^s )$ for all $t,s,r$ such that $t \neq s$ and $s \neq r$ (note $t = r$ is allowed). It is easy to see that (a) $\argmin_{U \in \mathcal{U}_{WB}} A(U;D)$ is empty for all $D \in \mathcal{D}'$ (use Proposition \ref{prop:AEI} for this), (b) $\mathcal{D}'$ is open, and (c) $\mathcal{D}'$ is dense in $\mathcal{D}$. It follows that the datasets $D \in \mathcal{D}$ for which $\argmin_{U \in \mathcal{U}_{WB}} A(U;D)$ is non-empty is rare.
\end{proof}

\begin{proof}[Proof of Proposition \ref{prop:uniform-pref}.]
	It is obvious that $\mathcal{R}_Q$ is reflexive. To see that it is transitive, suppose that $\bs{q} \ \mathcal{R}_Q \ \bs{q}'$ and $\bs{q}' \ \mathcal{R}_Q \ \bs{q}''$. This means there exists $\ti{U}$ and $\ti{U}'$ in $\mathcal{U}$ so that $U( \bs{q} ) \geq U(\bs{q}')$ for all $U \in \mathcal{U}$ satisfying $Q(U;D) \leq Q(\ti{U};D)$ and $U(\bs{q}') \geq U(\bs{q}'')$ for all $U \in \mathcal{U}$ satisfying $Q(U;D) \leq Q(\ti{U}';D)$. Let $\ti{U}'' \in \mathcal{U}$ be any utility function which satisfies $Q(\ti{U}'';D) \leq Q(\ti{U};D)$ and $Q(\ti{U}'';D) \leq Q(\ti{U}';D)$ (in particular, we can let $\ti{U}''$ be the best-fitting utility function in $\{ \ti{U}, \ti{U}' \}$). Now, $U(\bs{q}) \geq U(\bs{q}') \geq U(\bs{q}'')$ for all $U$ satisfying $Q(U;D) \leq Q(\ti{U}'';D)$ and therefore $\bs{q} \ \mathcal{R}_Q \ \bs{q}''$ and so $\mathcal{R}_Q$ is transitive. 

	Next, suppose that $U(\bs{q}) \geq U(\tbs{q})$ for all $U \in \mathcal{U}$. Let $\ti{U}$ be any element of $\mathcal{U}$. Clearly, for any $U$ which satisfies $Q(U;D) \leq Q(\ti{U};D)$ we have $U(\bs{q}) \geq U(\tbs{q})$ and so item 1 is proved. Next, suppose that $U(\bs{q}) \geq U(\tbs{q})$ implies $U(\bs{q}') \geq U(\tbs{q}')$ for all $U \in \mathcal{U}$ and suppose that $\bs{q} \ \mathcal{R}_{Q} \ \tbs{q}$. This means that there is some $\ti{U}$ so that $U(\bs{q} \geq U(\tbs{q}))$ for all $U$ such that $Q(U;D) \leq Q(\ti{U};D)$. It follows that $U(\bs{q}') \geq U(\tbs{q}')$ for all $U$ such that $Q(U;D) \leq Q(\ti{U};D)$ and so $\bs{q}' \ \mathcal{R}_{Q} \ \tbs{q}'$.
\end{proof}

\subsection{Proof of Proposition \ref{prop:afriat-welfare}} \label{sec:afriat-appendix}

For convenience let $a(U;D) = \sup\{ e \in [0,1]: D \text{ is } e\text{-rationalized by } U \}$. Of course, $A(U;D) = 1 - a(U;D)$. Further, from the definition of $\mathcal{R}_A$ we have $\tbs{q} \ \mathcal{R}_A \ \tbs{q}'$ if and only if there exists a well-behaved $\ti{U}$ so that $U(\tbs{q}) \geq U(\tbs{q}')$ for all well-behaved utility functions $U$ satisfying $a(U;D) \geq a(\ti{U};D)$.
\begin{lemma} \label{lemma:e-garp-to-artificial-data}
	Let $e \in [0,1]$ and let $D = (\bs{q}^t, \bs{p}^t)_{t \leq T}$ be a purchase dataset which satisfies $e$-GARP. Let $R_e$ be the $e$-direct revealed preference relation for $D$. For each $t$ let $B_e^t = \{ \bs{q} \in \mathbb{R}_+^L: e \bs{p}^t \cdot \bs{q}^t \geq \bs{p}^t \cdot \bs{q} \} \cup \{\bs{q}^t \}$ and let $\ti{D}_e$ be the choice dataset $\ti{D}_e = ( \bs{q}^t, B_e^t )$ with direct revealed preference relation $\ti{R}_e$. The data $\ti{D}_e$ satisfies cyclical consistency and $\tran( \geq \cup R_e ) = \tran( \geq \cup \ti{R}_e )$.
\end{lemma}
\begin{proof}
	Let $P_e$ be the $e$ strict direct revealed preference relation for $D$ and let $\ti{P}_e$ be the strict direct revealed preference relation for $\ti{D}_e$. Suppose $\ti{D}_e$ does not satisfy cyclical consistency. This means there exists $t_1, t_2, \ldots, t_K$ so that $\bs{q}^{t_1} \ \ti{R}_e \ \bs{q}^{t_2} \ \ti{R}_e \ldots \ \ti{R}_e \bs{q}^{t_K}$ and $\bs{q}^{t_K} \ \ti{P}_e \ \bs{q}^{t_1}$. Note that if $\bs{q}^{t_k} \geq \bs{q}^{t_{k+1}}$ for some $k$ then $\bs{q}^{t_{k-1}} \ \ti{R}_e \ \bs{q}^{t_{k+1}}$ and so without loss of generality we may assume that $\bs{q}^{t_{k}} \ngeq \bs{q}^{t_{k+1}}$ for all $k$. But this implies $\bs{q}^{t_1} \ R_e \ \bs{q}^{t_2} \ R_e \ldots \ R_e \ \bs{q}^{t_K}$ and $\bs{q}^{t_K} \ P_e \ \bs{q}^{t_1}$ and so $D$ does not satisfy $e$-GARP. We conclude that if $D$ satisfies $e$-GARP then $\ti{D}_e$ satisfies cyclical consistency. The equality $\tran( \geq \cup R_e ) = \tran( \geq \cup \ti{R}_e )$ follows from the fact that $\ti{R}_e = (\geq \cup R_e)$. 
\end{proof}

\begin{lemma}  \label{lemma:afriat-e-garp1}
	Let $e \in [0,1]$ and let $D = (\bs{q}^t, \bs{p}^t)_{t \leq T}$ be a purchase dataset. There exists a well-behaved utility function $U$ which satisfies $a(U;D) \geq e$ if and only if $D$ satisfies $e$-GARP. 
\end{lemma}
\begin{proof}
	Suppose $U$ is a well-behaved utility function which satisfies $a(U;D) \geq e$. From the definition of $a(U;D)$ we see that $U(\bs{q}^t) \geq U(\bs{q})$ for all $\bs{q} \in \rlp$ satisfying $e \bs{p}^t \cdot \bs{q}^t \geq \bs{p}^t \cdot \bs{q}$ and therefore $\bs{q}^t \ R_e \ (P_e) \ \bs{q}$ implies $U(\bs{q}^t) \geq (>) \ U(\bs{q})$. From this it is clear that $D$ satisfies $e$-GARP.
	
	For the ``if'' part of the proof let $\ti{D}_e$ denote the choice dataset in Lemma \ref{lemma:e-garp-to-artificial-data}. As $\ti{D}_e$ satisfies cyclical consistency (we know this from Lemma \ref{lemma:e-garp-to-artificial-data}) we may use the NOQ Theorem (Lemma \ref{lemma:NOQ}) to see that there exists a well-behaved utility function $U$ which rationalizes $\ti{D}_e$. This utility function satisfies $U(\bs{q}^t) \geq U(\bs{q})$ for all $\bs{q} \in \rlp$ satisfying $e \bs{p}^t \cdot \bs{q}^t \geq \bs{p}^t \cdot \bs{q}$ and therefore $a(U;D) \geq e$. 
\end{proof}

\begin{lemma} \label{lemma:afriat-e-garp2}
	Let $e \in [0,1]$ and let $D = (\bs{q}^t, \bs{p}^t)_{t \leq T}$ be a purchase dataset which satisfies $e$-GARP. For any two bundles $\tbs{q}$ and $\tbs{q}'$ we have $\tbs{q} \ \tran(\geq \cup R_e) \ \tbs{q}'$ if and only if $U( \tbs{q} ) \geq U(\tbs{q}')$ for all well-behaved $U$ satisfying $a(U;D) \geq e$. 
\end{lemma}
\begin{proof}
	Note that if $U$ is a well-behaved utility function satisfying $a(U;D) \geq e$ then clearly $\bs{q}^t \ R_e \ \bs{q}$ implies $U(\bs{q}^t) \geq U( \bs{q} )$ and so the ``only if'' part of the proof is straightforward. So, suppose that it is not the case that $\tbs{q} \ \tran(\geq \cup R_e) \ \tbs{q}'$. Let $\ti{D}_e$ be the choice dataset defined in Lemma \ref{lemma:e-garp-to-artificial-data} and let $\ti{R}_e$ be the direct revealed preference relation for $\ti{D}_e$. From Lemma \ref{lemma:e-garp-to-artificial-data} we see that it is not the case that $\tbs{q} \ \tran(\geq \cup \ti{R}_e) \ \tbs{q}'$. For each $\varepsilon>0$ let $\ti{D}_e^{\varepsilon}$ be the $T+1$ observation choice dataset formed by appending the observation $( \tbs{q}', \tbs{q} + (\varepsilon, \varepsilon,\ldots,\varepsilon) )$ to $\ti{D}_e$. Let $\ti{R}_e^{\varepsilon}$ and $\ti{P}_e^{\varepsilon}$ be the direct and strict direct revealed preference relations for $\ti{D}_e^{\varepsilon}$. By selecting $\varepsilon > 0$ sufficiently small we may ensure that for all $t$ we have $\tbs{q}' \ \ti{R}_{e}^{\varepsilon} \ \bs{q}^t$ implies either $\tbs{q} \geq \bs{q}^t$ or $\tbs{q}' \geq \bs{q}^t$. We claim that $\ti{D}_{e}^{\varepsilon}$ satisfies cyclical consistency. For a contradiction suppose there is a $t_1, t_2, \ldots, t_K$ satisfying $\bs{q}^{t_1} \ \ti{R}_e^{\varepsilon} \ \bs{q}^{t_2} \ \ti{R}_{e}^{\varepsilon} \ \ldots \ti{R}_{e}^{\varepsilon} \ \bs{q}^{t_K}$ and $\bs{q}^{t_K} \ \ti{P}_e^{\varepsilon} \ \bs{q}^{t_1}$. As $\ti{D}_e$ satisfies cyclical consistency there must be some $k$ so that $\bs{q}^{t_k} = \tbs{q}'$ and $\tbs{q} \geq \bs{q}^{t_{k+1}}$. From this it is clear that $\tbs{q} \ \tran(\geq \cup \ti{R}_e) \ \tbs{q}'$ but we have already shown that this is not the case and so we have arrived at a contradiction. We conclude that $\ti{D}_e^{\varepsilon}$ satisfies cyclical consistency and so by the NOQ Theorem (Lemma \ref{lemma:NOQ}) there exists a well-behaved utility function $\ti{U}$ which rationalizes $\ti{D}_e^{\varepsilon}$. Clearly $\ti{U}(\tbs{q}') > \ti{U}(\tbs{q})$ and $\ti{U}(\bs{q}^t) \geq \ti{U}( \bs{q} )$ for all $\bs{q} \in \rlp$ satisfying $e \bs{p}^t \cdot \bs{q}^t \geq \bs{p}^t \cdot \bs{q}$ and so $a(U;D) \geq e$. Thus, we see it is not the case that $U( \tbs{q} ) \geq U(\tbs{q}')$ for all well-behaved $U$ satisfying $a(U;D) \geq e$ and the proof is complete.
\end{proof}

\begin{proof}[Proof of Proposition \ref{prop:afriat-welfare}.]
	Suppose $D$ satisfies $e^*$-GARP. From Lemma \ref{lemma:afriat-e-garp1} there exists a well-behaved $\ti{U}$ satisfying $a(U;D) = e^*$. From Lemma \ref{lemma:afriat-e-garp2} we see that $\tbs{q} \ \tran(\geq \cup R_{e^*}) \ \tbs{q}'$ if and only if $U(\tbs{q}) \geq U(\tbs{q}')$ for all well-behaved $U$ satisfying $A(U;D) \leq A(\ti{U};D)$ and so indeed we have $\tbs{q} \ \tran(\geq \cup R_{e^*}) \ \tbs{q}'$ if and only if $\tbs{q} \ \mathcal{R}_A \ \tbs{q}'$. Next, let us suppose that $D$ does not satisfy $e^*$-GARP. 
	
	Suppose that $\tbs{q} \ \tran( \geq \cup \ P_{e^*} ) \ \tbs{q}'$. Let $\ti{e} < e^*$ be some number chosen large enough so that $\tbs{q} \ \tran(\geq \cup R_{\ti{e}} ) \ \tbs{q}'$. As $D$ satisfies $\ti{e}$-GARP Lemma \ref{lemma:afriat-e-garp1} shows that there exists a well-behaved $\ti{U}$ satisfying $a(\ti{U};D) \geq \ti{e}$. Lemma \ref{lemma:afriat-e-garp2} shows that for any well-behaved utility function $U$ satisfying $a(U;D) \geq a(\ti{U};D)$ we have $U(\tbs{q}) \geq U(\tbs{q}')$ and so indeed $\tbs{q} \ \mathcal{R}_A \ \tbs{q}'$. 
	
	Next, suppose that it is not the case that $\tbs{q} \ \tran( \geq \cup \ P_{e^*} ) \ \tbs{q}'$, let $\ti{U}$ be any well-behaved utility function, and let $\ti{e} = a(\ti{U};D)$. Because $D$ does not satisfy $e^*$-GARP Lemma \ref{lemma:afriat-e-garp1} shows that $e^* > \ti{e}$. Because it is not the case that $\tbs{q} \ \tran( \geq \cup \ P_{e^*} ) \ \tbs{q}'$ it is also not the case that $\tbs{q} \ \tran( \geq \cup \ R_{\ti{e}} ) \ \tbs{q}'$ and therefore Lemma \ref{lemma:afriat-e-garp2} guarantees that there exists a well-behaved $U$ so that $a(U;D) \geq \ti{e}$ and $U(\tbs{q}') > U(\tbs{q})$. As $\ti{U}$ was taken to be any well-behaved utility function we see that indeed $\tbs{q} \ \cancel{\mathcal{R}_A} \ \tbs{q}'$ as required.
\end{proof}

\subsection{Exact Afriat Index Computation} \label{sec:e^*-calc}

In Proposition \ref{prop:afriat-welfare} we defined $e^*$ as
\begin{equation*}
	e^* = \sup \{ e \in [0,1]: D \text{ satisfies } e\text{-GARP} \}
\end{equation*}
and so $A(D) = 1-e^*$. As noted, it is common in empirical applications to calculate an approximation of $e^*$ by performing a binary search over the interval $[0,1]$. Here we show how $e^*$ can be calculated exactly. 

Recall that $\bs{q}^t \ P_e \ \bs{q}^s$ means $e \bs{p}^t \cdot \bs{q}^t > \bs{p}^t \cdot \bs{q}^s$. We say that $D$ is \emph{$e$-acyclic} if, for all $t_1, t_2, \ldots, t_K$, it is not the case that $\bs{q}^{t_1} \ P_e \ \bs{q}^{t_2} \ P_e \ \bs{q}^{t_3} \ P_e \ \ldots \ P_e \ \bs{q}^{t_K} \ P_e \ \bs{q}^{t_1}$. Define a set $G \subseteq [0,1]$ by
\begin{equation*}
	G = \left\{ \dfrac{ \bs{p}^t \cdot \bs{q}^s }{ \bs{p}^t \cdot \bs{q}^t } \in [0,1]: t,s \in \{1,2,\ldots, T\} \right\}
\end{equation*}
Clearly, $G$ has no more than $T^2$ elements. Note that a dataset $D$ which satisfies $e$-GARP is $e$-acyclic but the reverse implication need not hold. However, the next result shows that the largest value of $e \in G$ at which $D$ is $e$-acyclic coincides with $e^*$. 
\begin{proposition} \label{prop:e^*-formula}
	For any purchase dataset $D = (\bs{q}^t, \bs{p}^t)_{t \leq T}$ we have
	\begin{equation} \label{eq:e^*-formula}
		e^* = \max\{ e \in G: D \text{ is $e$-acyclic} \}
	\end{equation}
\end{proposition}
\begin{proof}
	Let $\ti{e} = \max\{ e \in G: D \text{ is } e\text{-acyclic} \}$ and let $\ti{e}' = \sup\{ e \in [0,1]: D \text{ is } e\text{-acyclic} \}$. We first show that $\ti{e} = \ti{e}'$. To this end, let us enumerate the elements of $G$ as $e_1 > e_2 > \ldots > e_N$ and let $e_{N+1} = 0$. Let $\mathcal{M} = \{ ( e_{k+1}, e_{k}] \subseteq [0,1]: k \in \{1,2,\ldots, N\} \}$. From the definition of $P_e$ we see that, for any $M \in \mathcal{M}$, we have $P_{e} = P_{e'}$ for all $e, e' \in M$. It follows that the supremum in the definition of $\ti{e}'$ is obtained at some point in $G$ and thus  $\ti{e} = \ti{e}'$.
	
	As every dataset which satisfies $e$-GARP is $e$-acyclic we see that $e^* \leq \ti{e}'$ and as we have just seen that $\ti{e} = \ti{e}'$ we conclude that $e^* \leq \ti{e}$. So, to complete the proof we show that $e^* \geq \ti{e}$.
	
	Clearly if $D$ satisfies $\ti{e}$-GARP then $e^* \geq \ti{e}$ so let us suppose that $D$ does not satisfy $\ti{e}$-GARP. Thus, there must be some $t_1, t_2, \ldots, t_K$ so that $\bs{q}^{t_1} \ P_{\ti{e}} \ \bs{q}^{t_2} \ P_{\ti{e}} \ \bs{q}^{t_3} \ P_{\ti{e}} \ \ldots \ P_{\ti{e}} \ \bs{q}^{t_K}$, $\bs{q}^{t_K} \ R_{\ti{e}} \ \bs{q}^{t_1}$. As $D$ is $\ti{e}$-acyclic it must be that $\bs{q}^{t_K} \ \cancel{P_{\ti{e}}} \ \bs{q}^{t_1}$ or, in other words, $\ti{e} \bs{p}^{t_K} \cdot \bs{q}^{t_K} = \bs{p}^{t_K} \cdot \bs{q}^{t_1}$. Therefore, for any $e < \ti{e}$ we have $\bs{q}^{t_K} \ \cancel{R_e} \ \bs{q}^{t_1}$. From this it follows that $D$ satisfies $e$-GARP for any $e < \ti{e}$ and therefore $e^* = \ti{e}$.
\end{proof}
From Proposition \ref{prop:e^*-formula} we see that we can use $\max\{ e \in G: D \text{ is } e\text{-acyclic} \}$ as a formula for calculating $e^*$. Thus, $e^*$ can be computed by determining if $D$ is $e$-acyclic or not (this can be accomplished with Warshall's algorithm) for each $e \in G$ and then taking the largest value of $e$ for which $D$ is $e$-acyclic. This approach is computationally feasible because $G$ has at most $T^2$ elements. In general, one does not have to work out whether $D$ is $e$-acyclic for each $e \in G$. One could start from the largest element in $G$ and work their way down until a value of $e$ was found for which $D$ is $e$-acylic.

\bibliographystyle{myplainnat}
\bibliography{EconReferences}

\end{document}